\documentclass[aps,pra,amsmath,amssymb,twocolumn,superscriptaddress,nofootinbib]{revtex4-1}
\pdfoutput=1

\usepackage[dvipsnames]{xcolor}
\usepackage{graphicx}
\usepackage{dcolumn,bm,hyperref,amsmath,verbatim,amssymb,physics,amsthm,dsfont} 
\usepackage{mathtools}
\usepackage{enumerate}
\usepackage{thmtools}
\usepackage{thm-restate}

\mathtoolsset{showonlyrefs}

\newtheorem{theorem}{Theorem}
\newtheorem{lemma}[theorem]{Lemma}
\newtheorem{corollary}[theorem]{Corollary}
\newtheorem{definition}{Definition}
\newtheorem{proposition}[theorem]{Proposition}

\newcommand{\smsv}{\text{sq-sm}}
\newcommand{\tmsv}{\psi}

\newcommand{\tkappa}{\tilde{\kappa}}
\newcommand{\conv}{\operatorname{conv}}
\newcommand{\hil}{\mathcal{H}}
\newcommand{\convG}{\operatorname{conv}(\mathcal{G})}

\newcommand{\rmd}{\mathrm{d}}
\newcommand{\id}{\mathds{1}}

\begin{document}

\preprint{APS/123-QED}

\title{Resource distillation in convex Gaussian resource theories}

\author{Hyejung H. Jee}
\email{h.jee17@imperial.ac.uk}
\affiliation{Department of Computing, Imperial College London, London SW7 2AZ, UK}
\author{Carlo Sparaciari}
\affiliation{Department of Computing, Imperial College London, London SW7 2AZ, UK}
\affiliation{Department of Physics and Astronomy, University College London, London WC1E 6BT, UK}
\author{Mario Berta}
\affiliation{Department of Computing, Imperial College London, London SW7 2AZ, UK}


\begin{abstract}
It is known that distillation in continuous variable resource theories is impossible when restricted to Gaussian states and operations. To overcome this limitation, we enlarge the theories to include convex mixtures of Gaussian states and operations. This extension is operationally well-motivated since classical randomness is easily accessible. We find that resource distillation becomes possible for convex Gaussian resource theories\,---\,albeit in a limited fashion. We derive this limitation by studying the convex roof extension of a Gaussian resource measure and then go on to show that our bound is tight by means of example protocols for the distillation of squeezing and entanglement.
\end{abstract}

\maketitle


\section{Introduction}\label{sec:intro}

The Gaussian framework is a powerful tool for exploring quantum theory due to its easy accessibility, both theoretically and experimentally. Some preferred physical systems for implementing quantum information processing tasks are therefore in the Gaussian platform, such as optical and optomechanical systems. Despite being a manageable and useful tool for quantum physicists, several limitations on information processing tasks are known to hold for the Gaussian framework. For instance, Gaussian quantum computation does not provide an advantage over classical computation \cite{lloyd1999quantum,bartlett2002universal,menicucci2006universal,ohliger2010limitations}, correcting Gaussian errors with Gaussian operations is impossible \cite{niset2009no}, and distilling entanglement from Gaussian states using Gaussian operations is impossible \cite{fiuravsek2002gaussian,giedke2002characterization,eisert2002distilling}. This latter limitation is a particularly crucial one since the most natural platform for quantum communications and cryptography is a photonic system, and distillation of pure entangled photons is indispensable.

Recently, Gaussian resource theories were formally discussed in \cite{lami2018gaussian}, where the authors study quantum resource theories which are further restricted to Gaussian states and operations. It is shown that resource distillation is generally impossible in such Gaussian resource theories. A natural question to ask is whether there is a simple way to overcome this obstacle. One possibility is to introduce non-Gaussian resources, which then allows for resource distillation~\cite{browne2003driving} as well as other tasks~\cite{lloyd1999quantum,bartlett2002universal,gottesman2001encoding,niset2009no}, but these resources are generally difficult or expensive to produce.


\section{Overview of results}

In this paper, we examine whether classical randomness and post-selection can lift the restrictions on resource distillation in Gaussian resource theory. Specifically, we introduce the notion of \emph{convex Gaussian resource theories}. That is, theories obtained from standard Gaussian resource theories by adding the possibility of mixing Gaussian states and performing operations conditioned on the outcome of Gaussian measurements. The inclusion of the above operations is well-motivated from an operational point of view since, in a laboratory, it is not difficult to create a mixture of quantum states using classical randomness, or to condition on the outcome of a measurement. Our approach is partly inspired by the non-Gaussianity distillation protocol developed in, \cite{takagi2018convex} where convex mixtures of Gaussian states and operations play an important role.

Within this framework, we investigate resource distillations to see whether the restrictions affecting Gaussian resource theories can be overcome. Interestingly, we find that the restriction on resource distillation can be relaxed in convex Gaussian resource theories. The impossibility of resource distillation is now replaced by a limitation on the amount of distillable resources in convex Gaussian resource theories. This still leaves some room for resource distillation and we explore this boundary by constructing resource distillation protocols exploiting convex mixtures of Gaussian states and operations. These examples show that our derived upper limit on the amount of distillable resources can be saturated in some cases\,---\,which implies that it is tight. The main tool we use to prove the new limitation on resource distillation in convex Gaussian resource theories is the convex roof extension of a Gaussian resource measure introduced in \cite{lami2018gaussian}. Additionally, we introduce a multi-resource perspective on Gaussian resource theories.

The rest of the paper is structured as follows. In Sec.~\ref{sec:AddConvToGRT}, we revise the Gaussian quantum resource theories defined in \cite{lami2018gaussian} from the multi-resource theoretical point of view and introduce the concept of convex Gaussian resource theories. In Sec.~\ref{sec:general}, we derive an upper bound on the amount of distillable resources in convex Gaussian resource theories. In Sec.~\ref{sec:examples}, we provide explicit example protocols showing that resource distillation becomes possible in a limited fashion in convex Gaussian resource theories. We conclude in Sec.~\ref{sec:discussion}.


\section{Convex Gaussian resource theories}\label{sec:AddConvToGRT}

\subsection{Quantum resource theories}

A quantum resource theory is a mathematical framework designed to study the role of different physical quantities in physically constrained settings, where these quantities become resourceful for a given task~\cite{chitambar2019quantum}. The aim is to quantify the resource in a quantum state and analyse its change under different sets of quantum operations. Quantum resource theories have been successfully applied to many physical quantities of interest in the quantum information field such as entanglement theory \cite{horodecki_quantum_2009}, the resource theory of magic \cite{veitch2014resource,ahmadi2018quantification}, or the resource theory of asymmetry \cite{gour2008fundamental,lostaglio2015quantum}. There are three key ingredients in a quantum resource theory: the \emph{resource}, the \emph{free states} $\mathcal{F}$, and the \emph{allowed operations} $\mathbb{O}$. The resource is a physical quantity that is difficult or expensive to produce under some restrictions. The free states are states which do not possess any of the resource, and they can be produced for free. Any state which is not in $\mathcal{F}$ is regarded as resourceful. The definition of allowed operations varies among the literature; they are in general defined as operations that cannot create any resource from the free states, but are sometimes further restricted by physical constraints. For example, in entanglement theory, the largest set of operation which preserves the free states (the separable states) is the set of separable operations. However, the conventional set of allowed operations is composed by local operations and classical communications (LOCC), which is a subset of the set of separable operations.


\subsection{Gaussian resource theories}\label{subsec:GaussianQRT}

\paragraph{Setting.} A Gaussian resource theory is a quantum resource theory further restricted to Gaussian states and operations \cite{lami2018gaussian}. For example, the Gaussian resource theory of entanglement is the entanglement theory restricted to those systems described by Gaussian states with allowed operations given by Gaussian LOCC. The aim of Gaussian resource theories is to investigate capabilities and limitations of the Gaussian toolbox under given physical restrictions. This is a useful framework as many quantum information protocols have been proposed and implemented on Gaussian platforms due to their easy accessibility: quantum teleportation using Gaussian states \cite{vaidman1994teleportation,braunstein1998teleportation,furusawa1998unconditional}, bosonic quantum cryptography \cite{cerf2001quantum,grosshans2002continuous,grosshans2003quantum}, Gaussian error correction \cite{niset2009no}, etc. 

Here, we formulate Gaussian resource theories in a slightly different way from \cite{lami2018gaussian}, namely as an instance of multi-resource theories \cite{sparaciari2020first}. A Gaussian resource theory can be seen as a multi-resource theory with two resources. These two resources respectively form a separate resource theory; $\mathcal{R}$ is the theory associated with the resource $f$, while $\mathcal{R}_G$ is the resource theory of non-Gaussianity. The resource theory $\mathcal{R}$ has a set of free states $\mathcal{F}$ and a set of allowed operations $\mathbb{O}$.

In this paper, we only consider resource theories whose set of free states $\mathcal{F}$ satisfies the following assumptions. These are mostly standard assumptions about the set of free states in quantum resource theories formalised in \cite{brandao2015reversible}, except the displacement assumption additionally introduced for the Gaussian setting in \cite{lami2018gaussian}.

\noindent Namely, the set $\mathcal{F}$ is
\begin{enumerate}[(I)]
\item invariant under displacement operations
\item closed under tensor products, partial traces, and permutations of subsystems
\item convex
\item norm-closed.
\end{enumerate}
A prominent example that does not satisfy the above conditions is the resource theory of thermodynamics \cite{brandao2015second,serafini2017quantum}, that fails to satisfy property (I).

For the resource theory of non-Gaussianity $\mathcal{R}_G$, the set of free states is the set of Gaussian states $\mathcal{G}$, and the allowed operations are Gaussian operations $\mathbb{G}$. A Gaussian resource theory can then be built from these two resource theories, $\mathcal{R}$ and $\mathcal{R}_G$. The new free states are free Gaussian states denoted by $\mathcal{F}_G$ which are Gaussian states without any resource, i.e.~$\mathcal{F}_G=\mathcal{F}\cap\mathcal{G}$ as depicted in left panel of Fig.~\ref{fig:GaussianRT}. Since the resourcefulness of Gaussian states depends on their covariance matrices alone due to property I, $\mathcal{F}_G$ can be alternatively represented by the set of covariance matrices of the free Gaussian states, which we denote as $\mathbf{F}_G$. The set of allowed operations of the Gaussian resource theory are operations mapping $\mathcal{F}_G$ to itself. That is, operations that do not create neither resource $f$ nor non-Gaussianity from free Gaussian states. We refer to these operations as \emph{Gaussian allowed operations}, and the corresponding set is denoted by $\mathbb{O}_G$. Note that this set might include probabilistic maps as well as deterministic maps.


\paragraph{No Gaussian resource distillation.} As the no-go statement for Gaussian resource distillation derived in \cite{lami2018gaussian} is important for our results, we introduce it here in more detail. For Gaussian resource theories one can show that $\mathbf{F}_G$, the set of free covariance matrices, is \emph{upward closed}. That is, for any covariance matrices $V$ and $W$, if $V\in\mathbf{F}_G$ and $W\geq V$, then $W\in\mathbf{F}_G$. By exploiting the property of upward closeness, one can introduce a useful resource measure for a Gaussian resource theory with $\mathbf{F}_G$,
\begin{align}\label{eq:kappa}
\kappa\Big(\rho^G[V]\Big) := \min\Big\{t\geq1 \,\big| \,tV\in\mathbf{F}_G\Big\}\,,
\end{align}
where $\rho^G[V]$ is a Gaussian state with covariance matrix $V$. This measure is faithful in the sense that $\kappa(\rho^G)=1$ if and only if $\rho^G\in\mathcal{F}_G$ and furthermore has the property
\begin{align}
\text{$\kappa\left(\rho^G\otimes\sigma^G\right) = \max\Big\{\kappa\left(\rho^G\right),\kappa\left(\sigma^G\right)\Big\}$ for $\rho^G, \sigma^G\in\mathcal{G}$.}
\end{align}
This measure is a resource monotone for Gaussian resource theory satisfying Assumptions (I)--(IV). This implies the following no-go statement for resource distillation~\cite[Thm.~1]{lami2018gaussian}.

\begin{proposition}\label{thm:nogo_theorem}
For a Gaussian resource theory satisfying (I)--(IV) with
\begin{align}
\text{$\kappa\left(\rho^G\right)<\kappa\left(\sigma^G\right)$ for $\rho^G, \sigma^G\in\mathcal{G}$,}
\end{align}
there does not exist $\Gamma\in\mathbb{O}_G$ with $\Gamma\left(\left(\rho^G\right)^{\otimes n}\right)=\sigma^G$ for any $n\in\mathbb{N}$.
\end{proposition}

It is worth noting that Prop.~\ref{thm:nogo_theorem} includes {\it trace non-increasing} Gaussian allowed operations, meaning that even {\it probabilistic} resource distillation is impossible. This is a critical limitation as it rules out the possibility of resource distillation even with many copies of an initial state.

\begin{figure}[t]
\centering
\includegraphics[width=\columnwidth]{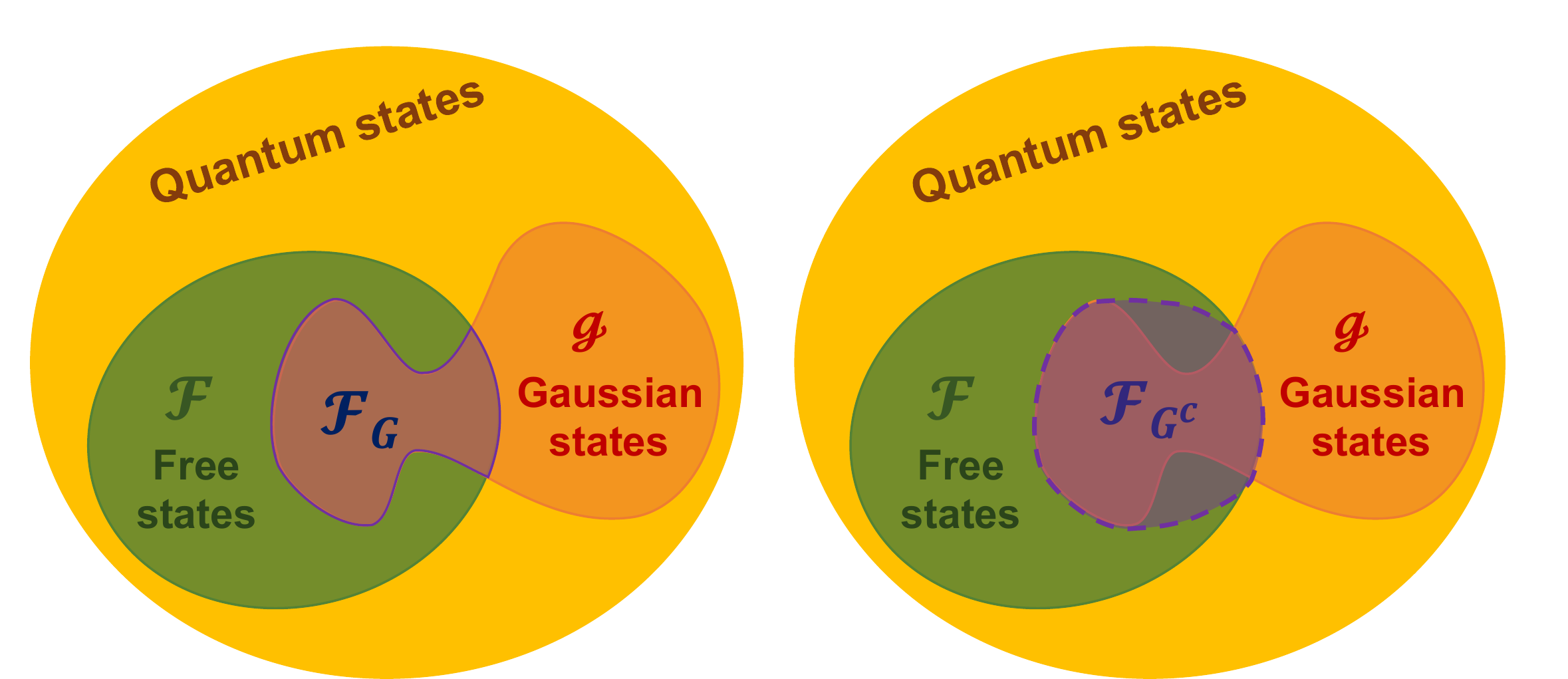}
\caption{The geometrical structures of the state spaces in resource theories. \textbf{Left}: In Gaussian resource theories, the set of free states is the interception of the original free state set $\mathcal{F}$ and the set of Gaussian states $\mathcal{G}$, i.e. $\mathcal{F}_G=$ $\mathcal{F}\cap\mathcal{G}$. \textbf{Right}: In convex Gaussian resource theories, the set of free states is the convex hull of $\mathcal{F}_G$, i.e. $\mathcal{F}_{G^c} = \conv(\mathcal{F}_G) = \conv(\mathcal{F}\cap\mathcal{G})$.}
\label{fig:GaussianRT}
\end{figure}


\subsection{Convex Gaussian resource theories}\label{subsec:convexGRT}

An important feature of the set of Gaussian states is that it is not convex since, in general, a classical mixture of Gaussian states is non-Gaussian. Although working with non-Gaussian mixtures of Gaussian states can be mathematically demanding, it is not difficult to create these states in the laboratory using Gaussian states and Gaussian operations. Indeed, we can obtain a convex mixture by performing probabilistic Gaussian operations over a system described by a Gaussian state. An example of probabilistic Gaussian operations is given by Gaussian operations conditioned on the outcome of measurements performed on other modes. Therefore, including non-Gaussian mixtures to the state space seems well-motivated from an operational point of view. It is also worth remarking that this inclusion has been studied in the context of the resource theory of non-Gaussianity in \cite{takagi2018convex}, and the convex resources of Gaussian operations turned out to be useful for distillation of non-Gaussianity.

In this section, we introduce \emph{convex Gaussian resource theories} which take into account the fact that creating convex mixtures of Gaussian states is not difficult in practice. We formulate these theories based on Gaussian resource theories defined in the last section. Now the state space of the theories is given by the convex hull of Gaussian states
\begin{align}
\conv(\mathcal{G}) = \left\{ \int \rmd \lambda \,p(\lambda) \, \rho^G(\lambda)
\ \Big\vert \ \rho^G(\lambda) \in \mathcal{G}  \right\}\,,
\end{align}
where $p(\lambda)$ is a probability distribution. The above set includes all Gaussian states as well as classical mixtures of Gaussian states.

To construct a convex Gaussian resource theory for a resource $f$, we start from the corresponding Gaussian resource theory of the resource $f$ with the set of free states $\mathcal{F}_G$ and the allowed operations $\mathbb{O}_G$. Then, the set of free states of the convex Gaussian resource theory, that we denote by $\mathcal{F}_{G^c}$ is given by the convex hull of $\mathcal{F}_G$. That is, $\mathcal{F}_{G^c}=\conv(\mathcal{F}_G)$ as depicted in the right panel of Fig.~\ref{fig:GaussianRT}. The set of allowed operations of the convex Gaussian resource theory, which we denote as $\mathbb{O}_{G^c}$, is defined as follows.

\begin{definition}\label{def:all_op_cgt}
The allowed operations $\mathbb{O}_{G^c}$ of a convex Gaussian resource theory are composed by two kinds of operations:
\begin{enumerate}
\item Appending an ancillary system described by a free state $\rho \in \mathcal{F}_{G^c}$.
\item Applying a mixture of Gaussian operations conditioned on the outcome of a homodyne measurement,
\begin{align}\label{eq:generalForm_convG}
\Gamma(\rho_{AB})=\int \rmd q \, \Phi_A^{q} \otimes M_B^{q}\left(\rho_{AB}\right)\,,
\end{align}
where the Gaussian operation $\Phi^{q} \in \mathbb{O}_{G}$ preserves the free states $\mathcal{F}_G$, and $M^{q} = \bra{q} \cdot \ket{q}$ is a selective homodyne measurement.
\end{enumerate}
\end{definition}

A hidden assumption we are making here is that the partial selective homodyne measurement $M^q$ needs to be an allowed operation too. That is, it needs to map the set $\mathcal{F}_{G}$ into itself. For the theories we consider explicitly, namely those of entanglement and squeezing, we show that this is indeed the case in App.~\ref{app:Part_sel_HM}. It is easy to see that the allowed operations in Def.~\ref{def:all_op_cgt} map the free set $\mathcal{F}_{G^c}$ into itself. Indeed, given a generic free state
\begin{align}
\rho_{\text{free}}=\int \rmd \lambda \, p(\lambda) \, \rho_{\text{free}}^G(\lambda)\,,
\end{align}
where $\rho_{\text{free}}^G(\lambda) \in \mathcal{F}_{G}$ for all $\lambda$, appending a free state (operation 1) does not map $\rho_{\text{free}}$ outside the set. On the other hand, a mixture of conditioned allowed operations, represented by $\Gamma$ (operation 2), is such that
\begin{align}\label{eq:allowed_ops_free_states}
\Gamma \left( \rho_{\text{free}} \right)=\int \rmd q \, \rmd \lambda \, p(\lambda) \, \Phi_A^{q} \otimes M_B^{q}\left( \rho_{\text{free}}^G(\lambda) \right)\,,
\end{align}
where we assume wlog that the Hilbert space of $\rho_{\text{free}}$ can be divided in two partitions $A$ and $B$. Since both $\Phi_A^{q}$ and $M_B^q \in \mathbb{O}_{G}$, each
state $\rho_{\text{free}}^G(\lambda)$ is mapped to another state in $\mathcal{F}_{G}$. Then, the outcome of Eq.~\eqref{eq:allowed_ops_free_states} is a mixtures of free states, which belongs to the set $\mathcal{F}_{G^c}$.

We can now investigate resource distillations in convex Gaussian resource theories. The main question we aim to address is whether distillation is still impossible if convexity is allowed in Gaussian resource theories. In other words, we study whether a version of Prop.~\ref{thm:nogo_theorem} holds for these theories or not. This question is equivalent to asking whether we can distil the resource $f$ from a resourceful Gaussian state $\rho$ in $\convG/\mathcal{F}_{G^c}$, by only using operations in $\mathbb{O}_{G^c}$. In the next section, we show that, if we are allowed to use classical mixtures of Gaussian states and operations, the restriction on resource distillation is weakened, and we precisely characterise these new limitations. Finally, it is worth remarking that the opposite case, distillation of non-Gaussianity from a state $\rho$ in $\mathcal{F}/\mathcal{F}_{G^c}$ only using operations in $\mathbb{O}_{G^c}$, has already been studied in \cite{takagi2018convex}. The authors showed that one can distil non-Gaussianity from any non-Gaussian state by filtering out the Gaussian part using post-selection. Their distillation protocol consists of a beam splitter and a homodyne measurement, which belong to $\mathbb{O}_{G^c}$ in most convex Gaussian resource theories.


\section{Limitations on resource distillation}\label{sec:general}

To study resource distillation in general convex Gaussian resource theories, we need an appropriate resource measure to work with. The measure defined in Eq.~\eqref{eq:kappa} is a good candidate, but unfortunately it cannot be applied to non-Gaussian convex mixtures of Gaussian states. Therefore, we define the following new resource measure for the convex hull of Gaussian states; for any $\rho\in\convG$, 
\begin{align}\label{eq:tkappa}
\tkappa(\rho) := \inf_{\substack{\{p(\lambda),\rho^G(\lambda)\}_{\lambda}\\ \rho=\int \rmd\lambda \,p(\lambda)\rho^G(\lambda)}} \int \rmd\lambda\, p(\lambda) \, \kappa\left(\rho^G(\lambda)\right),
\end{align}
where the infimum is over all Gaussian decompositions of the state $\rho$ including discrete ones. This measure is the \emph{convex roof extension} \cite{chitambar2019quantum} of the monotone $\kappa$ in Eq.~\eqref{eq:kappa}. It is a valid resource monotone satisfying faithfulness and monotonicity, as well as other useful properties listed in the following lemma.

\begin{restatable}{lemma}{proptk}
\label{lem:PropertiesTkappa}
The resource measure $\tilde{\kappa}$ defined in Eq.~\eqref{eq:tkappa} has the following properties:
\begin{enumerate}[(i)]
\item It coincides with the measure $\kappa$ for Gaussian states
\begin{align}
\tkappa(\rho^G) = \kappa(\rho^G) \quad \forall \, \rho^G\in\mathcal{G}\,.
\end{align}
\item It is \emph{convex}
\begin{align}
\tkappa
\left(\int d\lambda\, p(\lambda) \rho(\lambda)\right) 
\leq
\int d\lambda\, p(\lambda) \, \tkappa(\rho(\lambda))
\end{align}
for any probability distribution $p(\lambda)$ and set of states $\left\{ \rho(\lambda) \right\}_{\lambda}$ in $\convG$.
\item Its value on tensor products of Gaussian states is given by the most resourceful state
\begin{align}
\tilde{\kappa}\left(\rho^G\otimes\tau^G\right) = \max\Big\{\tilde{\kappa}\left(\rho^G\right),\tilde{\kappa}\left(\tau^G\right)\Big\}\,\quad\forall\rho^G,\tau^G\in\mathcal{G}\,.
\end{align}
\item It is monotonic under the allowed operations $\mathbb{O}_{G^c}$
\begin{align}
\tkappa(\rho)\geq\tkappa(\Gamma(\rho)) \quad \forall \, \Gamma\in\mathbb{O}_{G^c}\,.
\end{align}
\item It is \emph{faithful}
\begin{align}
\tkappa(\rho)=1 \iff \rho\in\mathcal{F}_{G^c}\,.
\end{align}
\end{enumerate}
\end{restatable}

The proof is given in App.~\ref{app:prop_tkappa_proof}. Using the resource measure $\tilde{\kappa}$ and its properties described in Lem.~\ref{lem:PropertiesTkappa}, we can prove the following theorem for resource distillation in convex Gaussian resource theories.  This theorem shows that no allowed operation can increase the value of the measure $\tkappa$ above that of the most resourceful state composing the initial mixture.

\begin{theorem}\label{thm:noGmap}
Let $\Gamma\in\mathbb{O}_{G^c}$ and $\rho=\int \rmd\lambda \, p(\lambda)\rho^G(\lambda) \in\conv(\mathcal{G})$. Then, we have for any $n$ that
\begin{align}
\tkappa\left(\Gamma\left(\rho^{\otimes n}\right)\right)\leq\tkappa\left(\rho^G_{\max}\right)\,,
\end{align}
where $\rho^G_{\max}$ denotes the most resourceful element in the set $\{\rho^G(\lambda)\}_{\lambda}$.
\end{theorem}

\begin{proof}
By exploiting monotonicity of the measure $\tkappa$  with respect to the allowed operations in $\mathbb{O}_{G^c}$, it follows that $\tkappa\left(\Gamma\left(\rho^{\otimes n}\right)\right)\leq\tkappa\left(\rho^{\otimes n}\right)$. For the given decomposition of $\rho$ we have
\begin{align}
\rho^{\otimes n} = \int \rmd \boldsymbol{\lambda} \, \bar{p}(\boldsymbol{\lambda}) \, \bar{\rho}^G(\boldsymbol{\lambda})\,,
\end{align}
where $\boldsymbol{\lambda} = \left( \lambda_1,\ldots,\lambda_n \right)$, the probability distribution is $\bar{p}(\boldsymbol{\lambda}) = \Pi_{i=1}^n \, p(\lambda_i)$, and the state is $\bar{\rho}^G(\boldsymbol{\lambda}) = \otimes_{i=1}^n \, \rho^G(\lambda_i)$. Then, it follows that
\begin{align}
\tkappa\left(\rho^{\otimes n}\right)
&\leq
\int \rmd \boldsymbol{\lambda} \, \bar{p}(\boldsymbol{\lambda}) \, \tkappa \left( \bar{\rho}^G(\boldsymbol{\lambda}) \right) \\
&\leq
\int \rmd \boldsymbol{\lambda} \, \bar{p}(\boldsymbol{\lambda}) \, \tkappa \left( \rho^G_{\max} \right)
=
\tkappa \left( \rho^G_{\max} \right)\,,
\end{align}
where the first inequality follows from convexity of $\tkappa$, property (ii), the second inequality from property (iii), and the state $\rho^G_{\max}$ is the one in $\left\{ \rho^G(\lambda) \right\}_{\lambda}$.
\end{proof}

A limitation for resource distillation in convex Gaussian resource theories naturally follows from Thm.~\ref{thm:noGmap}.

\begin{corollary}\label{cor:NewNGTheorem}
Let $\rho=\int \rmd\lambda\,p(\lambda)\rho^G(\lambda)\in\conv(\mathcal{G})$ with $\tau\in\convG$ such that $\tkappa(\tau)>\tkappa(\rho_{\max}^G)$ for the most resourceful element $\rho_{\max}^G$ in the set $\{\rho^G(\lambda)\}_{\lambda}$. Then, for any $n \in \mathbb{N}$ there does not exist $\Gamma\in\mathbb{O}_{G^c}$ such that $\Gamma\left(\rho^{\otimes n}\right)=\tau$.
\end{corollary}

Let us emphasise that this upper bound is only applicable to $\convG$, and it is not valid for general non-Gaussian states. Nevertheless, Cor.~\ref{cor:NewNGTheorem} states that, in convex Gaussian resource theories, one can never distil more resourceful states than those already present in the mixture\,---\,even with multiple copies of it. That is, distillation corresponds to the ability of identifying and intensifying a resource which is already present in the mixture. As such, Cor.~\ref{cor:NewNGTheorem} leaves a possibility of resource distillation that might be useful in some specific scenarios. For example, it does not rule out the possibilities for error correction or purification.

In the next section, we will show that one can find simple examples of resource distillation protocols in convex Gaussian resource theories. This examples will show that the upper bound on the amount of distillable resources characterised in Cor.~\ref{cor:NewNGTheorem} can be achieved in some cases.


\section{Convex distillation protocols}\label{sec:examples}

\subsection{Motivation}

Here, we provide explicit examples which show that we can still perform resource distillation in convex Gaussian resource theories. We start with the convex Gaussian resource theory of \emph{squeezing} and then go on to discuss \emph{entanglement}. We see the former as a useful toy model where our intuition for protocols can be built, while the latter corresponds to the most prominent resource theory.


\subsection{Squeezing distillation}

\paragraph{One-shot deterministic case.} As the aim is to demonstrate that resource distillation is possible in convex Gaussian resource theories, we can consider the simplest scenario first. Our scheme seems nonetheless fairly general and we expect it work for other more complicated settings as well. Consider a single-mode squeezed vacuum (SMSV) state with squeezing parameter $r$ along the $\hat{x}$ quadrature, which we denote as $\ket{0,r}$ (cf.~App.~\ref{app:Gaussian}). Suppose the system experiences a random displacement noise with probability $p$, so that its state becomes a mixture of a SMSV state and a displaced SMSV state. For simplicity, let us assume that the quadrature displaced is the same as the squeezed one. Then, the state can be written as
\begin{align}\label{eq:OS_SQ_initialstate}
\rho_{\text{in}}
=
(1-p) \, \ket{0,r}\!\!\bra{0,r}
+ p \, \ket{d,r}\!\!\bra{d,r}\,,
\end{align}
where $\ket{d,r} = D(d) \, \ket{0,r}$, and $D(d)$ is the operator displacing the state by $d$ along the $\hat{x}$ quadrature. For $d$ sufficiently large, the state $\rho_{\text{in}}$ does not have any squeezing that might be exploited for a quantum task (for metrology, for instance). Our aim is to distil the hidden squeezing resource from this mixture.

\begin{figure}[t!]
\centering
\includegraphics[width=0.6\columnwidth]{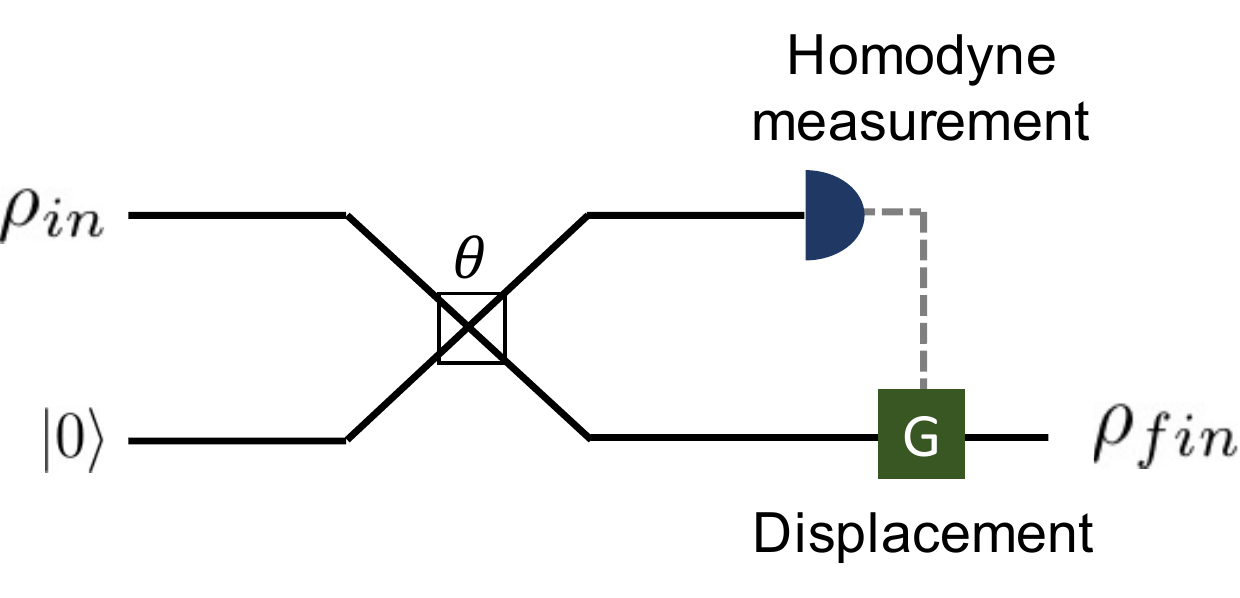}
\includegraphics[width=0.9\columnwidth]{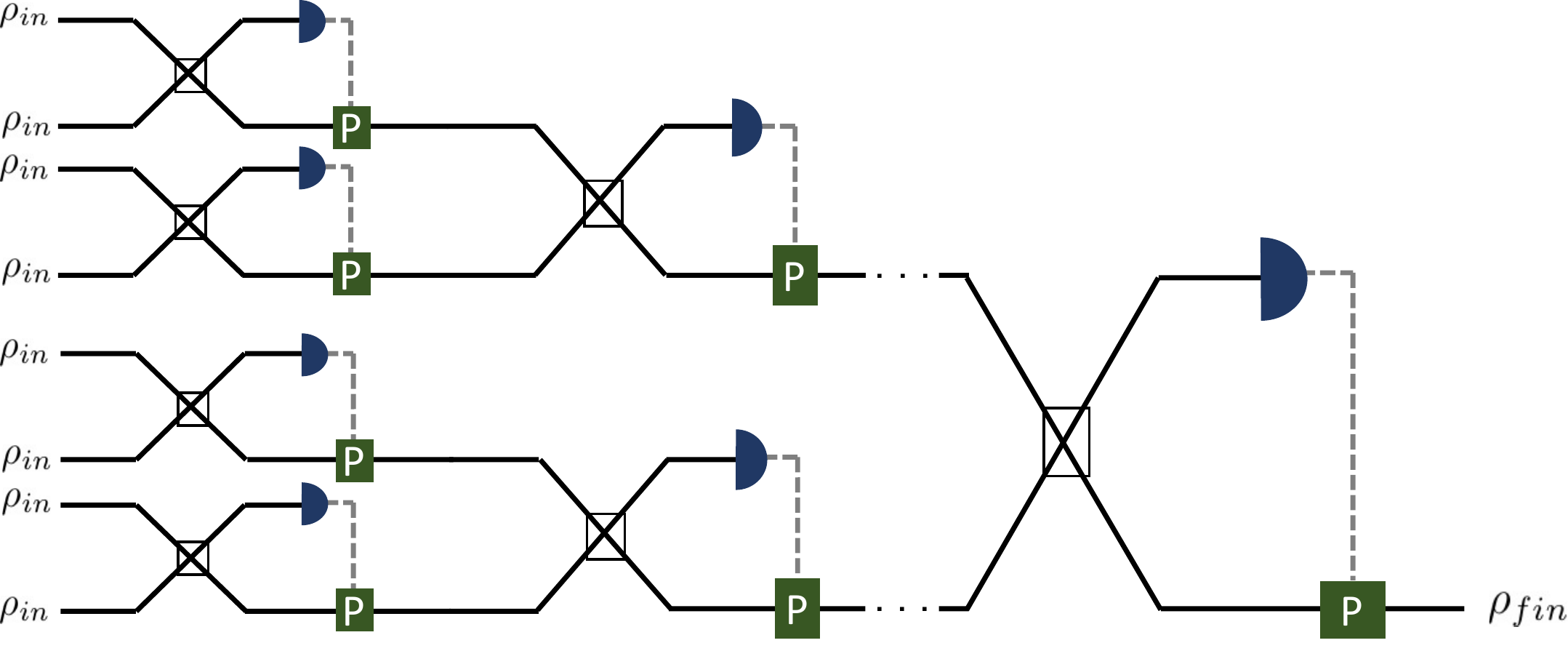}
\caption{{\bf Top}: Schematic of the deterministic one-shot squeezing distillation protocol. The system, described by the state $\rho_{\text{in}}$ in Eq.~\eqref{eq:OS_SQ_initialstate}, is fed into a beam splitter with transmissivity $t=\cos^2\theta$ together with a pointer system in the vacuum state. After the interaction, a homodyne measurement is performed on the pointer system and depending on the outcome the main system is displaced or not. {\bf Bottom}: Schematic of the probabilistic multi-copy squeezing distillation protocol. The pointer system is now replaced by another copy of $\rho_{in}$, providing better performance due to the fact that it is a mixture of squeezed states. The system and pointer interact via a beam splitter, and a homodyne measurement is performed on the pointer. If the outcome of the measurement falls within a given range, the protocol is successful, otherwise it is aborted. The protocol can be iterated multiple times, with the output system of one iteration being the input system of the next one.}
\label{fig:OneShot_SQdist}
\end{figure}

A simple distillation protocol, inspired by the one presented in \cite{heersink2006distillation}, is shown in the upper panel of Fig.~\ref{fig:OneShot_SQdist}. The main idea is to measure the $\hat{x}$ quadrature of the mixture, and correct the displacement noise accordingly. To do so, we first need to correlate the main system with a pointer, that we later measure, initialised in the vacuum state $\ket{0}$. The two systems interact via a beam splitter with large transmissivity $t = \cos^2\theta$. In this way, the pointer gets information on whether the system was displaced or not, while disturbing the state of the system and therefore reducing its squeezing. In particular, the larger the transmissivity of the beam splitter, the less the squeezing of the system is compromised, while at the same time the less is learned about the displacement affecting the system~\cite{hofmann_information_2000}.

After the beam splitter, the $\hat{x}$ quadrature of the pointer is measured via homodyne detection. In particular, the measurement is described by the dichotomic positive-operator valued measure (POVM) $\left\{ \Pi , \id - \Pi \right\}$, where the effect $\Pi$ is defined as
\begin{align}
\label{eq:sq_effect}
\Pi =
\int_{-\infty}^{\Delta} \rmd x \, \ket{x}\!\!\bra{x}\,,
\end{align}
where $\Delta = - \frac{d}{2} \sin \theta$. When the outcome of the POVM refers to $\Pi$, it means that with high probability the main system was displaced along $\hat{x}$, and we can correct it by performing a displacement of $-d \cos \theta$ on the system along the same quadrature. On the other hand, if the outcome of the POVM is associated with the effect $\id - \Pi$, with high probability the system was not displaced and we do not need to act on it.

\begin{figure}[ht!]
\centering
\includegraphics[width=0.9\columnwidth]{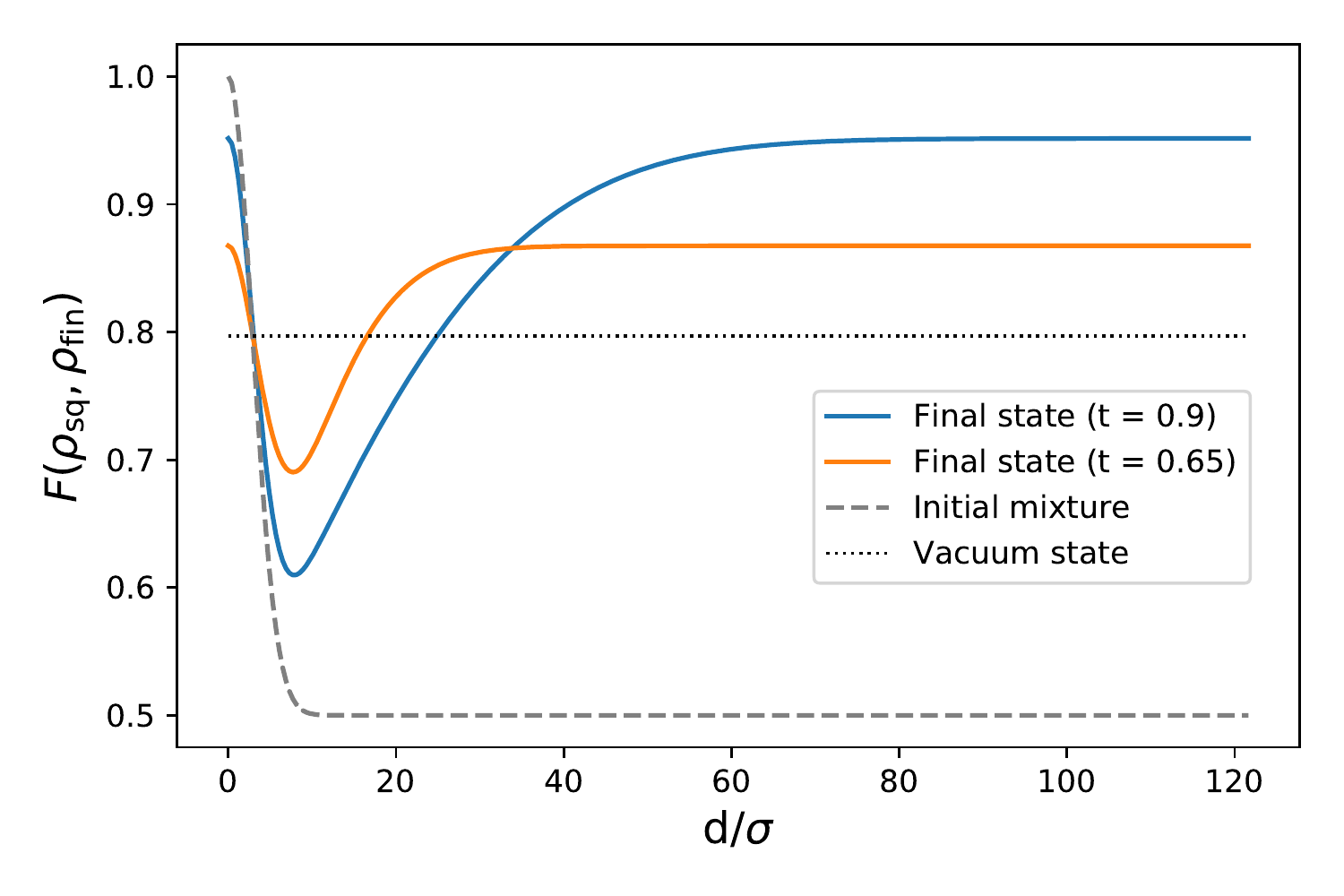}
\includegraphics[width=0.9\columnwidth]{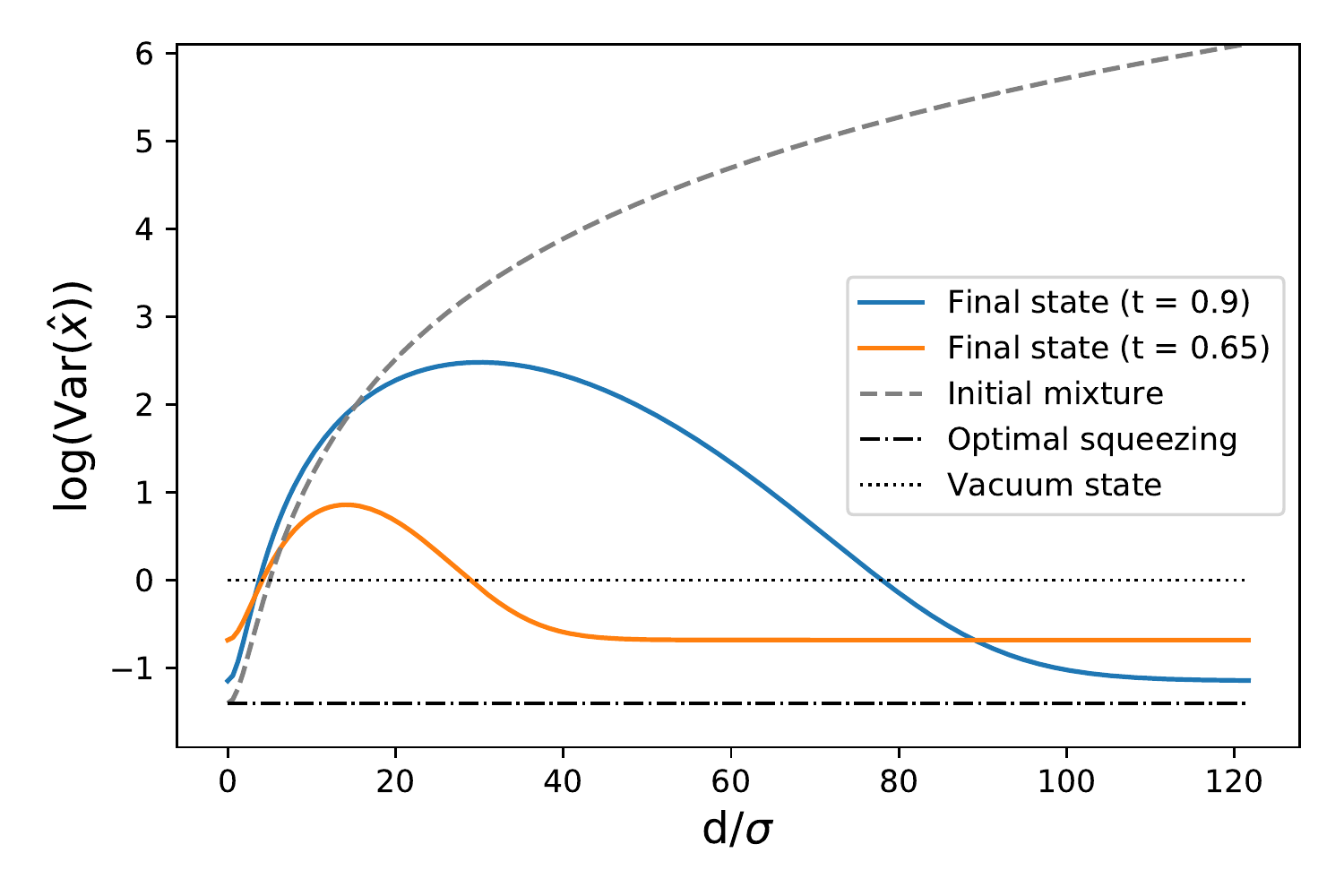}
\caption{Numerical simulation of the deterministic single-shot squeezing protocol, for two different beam splitter settings with transmissivity $t = 90\%$ and $65\%$. {\bf Top}: The fidelity between the desired state $\rho_{\text{sq}}=\ket{0,r}\!\!\bra{0,r}$ and the final state $\rho_{\text{fin}}$ of the protocol as a function of $d/\sigma$, which is the ratio between the amount of displacement $d$ randomly applied to the initial state and its $\hat{x}$-variance $\sigma=e^{-2r}$. While a low beam splitter's transmissivity $t=65\%$ allows for correcting smaller random displacement, it also degrades squeezing more consistently, thus achieving a lower maximum fidelity for high values of $d/\sigma$ than the one achieved by a higher transmissivity (e.g., $t=90\%$). {\bf Bottom}: The variance in the $\hat{x}$-quadrature for the final state of the protocol is shown as a function of $d/\sigma$. The protocol is successful when the final state has a variance lower than the one of the vacuum state. For the simulation we set the squeezing parameter $r=0.7$ and the probability of random displacement to $p=0.5$.}
\label{fig:numerics_sq_single}
\end{figure}

Numerical results, showing the performance of this protocol, are given in Fig.~\ref{fig:numerics_sq_single}. We characterise the performance in terms of the fidelity to the SMSV state $\ket{0,r}$ as a function of the distinguishability between the elements of the mixture. Distinguishability can be quantified by the ratio $d/\sigma$, where $d$ is the displacement along the relevant quadrature while $\sigma = e^{-2r}$ is the variance of $\ket{0,r}$ along the same quadrature. When $d\approx\sigma$, the protocol is not able to distinguish the two elements in Eq.~\eqref{eq:OS_SQ_initialstate}, and thus the outcome has a lower fidelity to the target state $\ket{0,r}$ than the initial one (dashed grey line). When instead $d\gg \sigma$, the protocol is able to distinguish the two states, and the displacement can be effectively corrected; the fidelity is closer to 1 than the initial state as well as the vacuum state (dotted grey line) which is the closest free state to $\ket{0,r}$. Also note that with larger transmissivity it takes longer to reach high fidelity as $d/\sigma$ increases, but higher fidelity is reached in the $d\gg \sigma$ region. This is because less information is extracted by the pointer from the initial state after the beam splitter, making it harder to distinguish the two states but preserving more squeezing.

Additionally, we compute the variance along the displaced quadrature (in this example the $\hat{x}$-quadrature), see the bottom panel in Fig.~\ref{fig:numerics_sq_single}. Variance is a relevant measure for this resource theory since squeezing is often considered as a resource for metrological tasks\,---\,for example to improve the sensitivity of interferometers\,---\,and variance quantifies how useful a given probe can be for these tasks. We find that our protocol can indeed reduce the variance of the initial mixture when $d \gg \sigma$, reaching values well below the ones of the vacuum.\footnote{In this theory, the vacuum is the free state with minimum variance in all quadratures.} However, we notice that the outcome state can never be more squeezed than the maximally-squeezed state in the original mixture (`optimal squeezing' in the graph), in accord with the upper bound of distillable resources derived in Cor.~\ref{cor:NewNGTheorem}. As a side-remark, we notice that variance is not a monotone for the convex resource theory of squeezing when the allowed operations are of the form given in Def.~\ref{def:all_op_cgt}. This is because conditional operations can reduce this measure. However, we show in App.~\ref{app:var_is_monotone} that variance is a meaningful monotone under a slightly different set of allowed operations, which is still natural over the convex hull of Gaussian states.

It is worth pointing out that we do not use the measure $\tkappa$ of Eq.~\eqref{eq:tkappa} in our numerical results, because the initial state and the target state $\ket{0,r}$ have the same value of $\tkappa$. Indeed, this measure is useful to derive general limitations on resource distillation in convex Gaussian resource theories, but it is not as good at reflecting the usefulness of the resource contained in a state. Note that the protocol here described is deterministic and requires only one copy of the initial state.


\paragraph{Multi-copy probabilistic case.} Within the setting of resource theories, having access to multiple copies of the same state often provides some benefit. For example, one can act globally over the copies to reduce quantum fluctuations, and consequently make the theory reversible, thus achieving better distillation or dilution rates~\cite{bennett_concentrating_1996}. Interestingly, this is in general not the case for Gaussian resource theory, as it was shown in \cite{lami2018gaussian}. In the following example, we explore whether the many-copy settings might provide advantages for distillation in the convex Gaussian theory of squeezing. We show that, indeed, with many copies it is possible to get information on the displacement noise affecting the system without compromising its amount of squeezing. Thus, under this point of view it seems that the multi-copy setting has some advantages over the single-shot case.

\begin{figure}[t!]
\includegraphics[width=0.9\columnwidth]{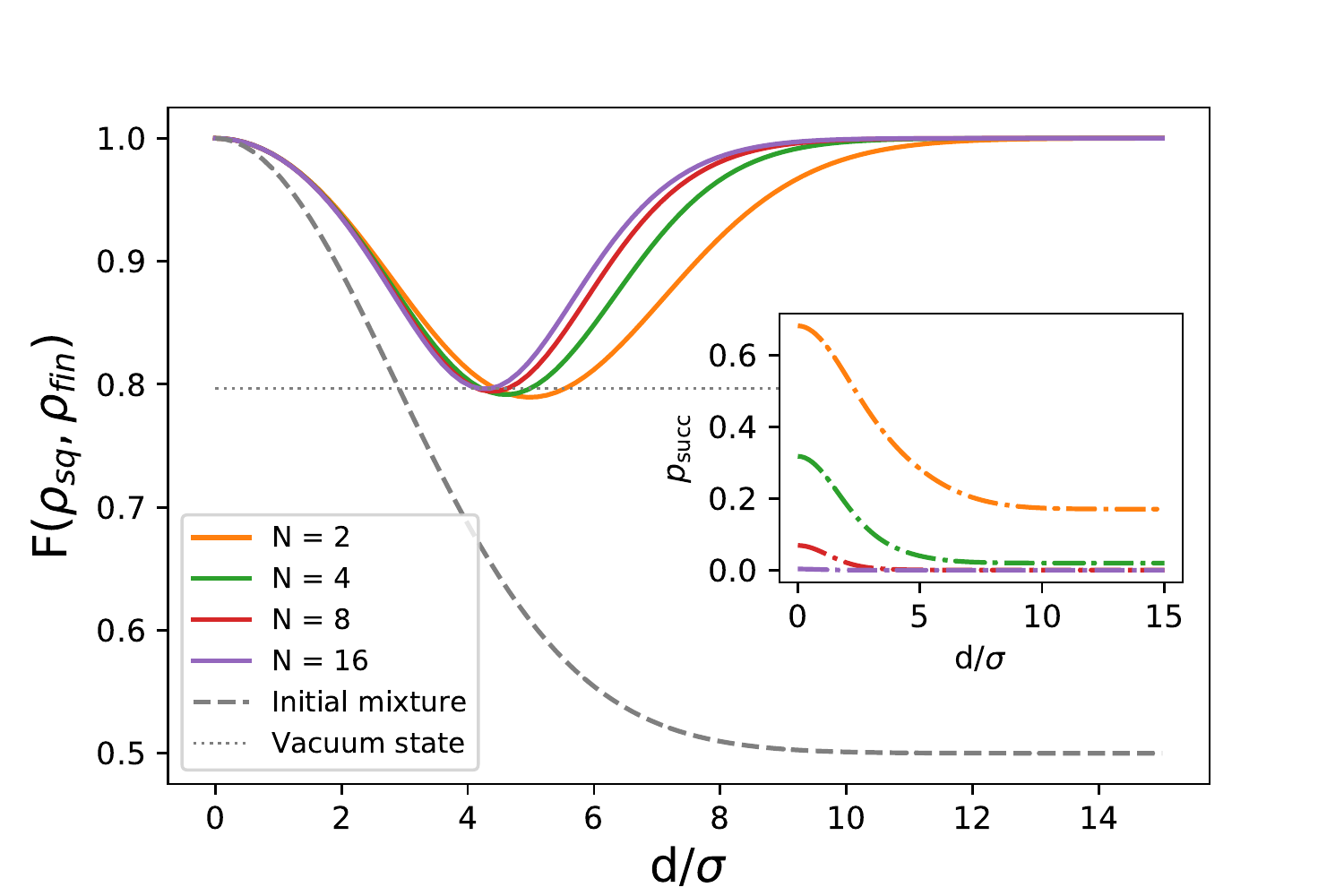}
\includegraphics[width=0.9\columnwidth]{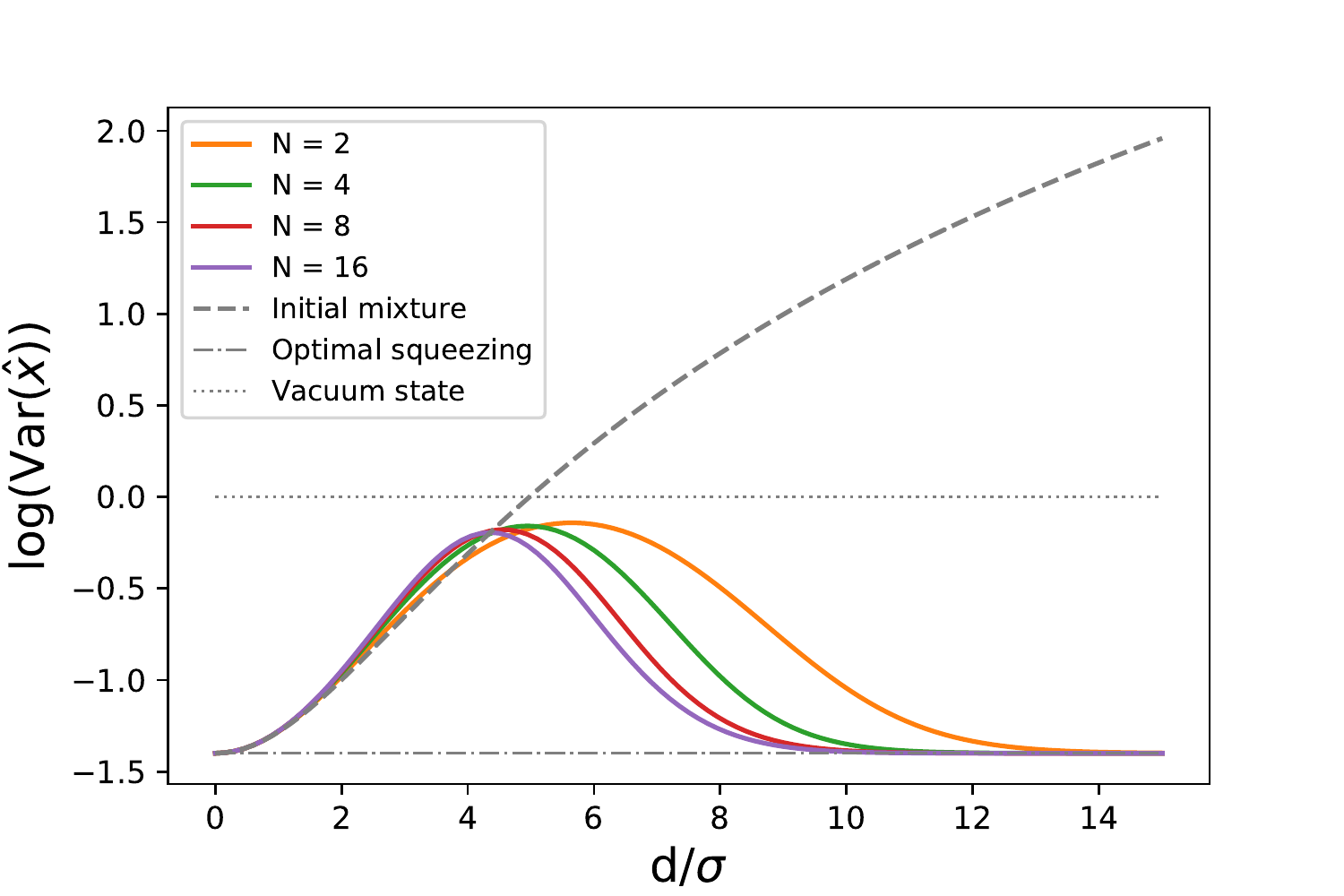}
\caption{Numerical simulation of the probabilistic multi-copy squeezing distillation protocol for different numbers of copies of the initial state. {\bf Top}: The fidelity between the desired state $\rho_{\text{sq}}=\ket{0,r}\!\!\bra{0,r}$ and the final state $\rho_{\text{fin}}$ of the protocol as a function of $d/\sigma$, where $\sigma=e^{-2r}$ is the $\hat{x}$-variance of $\ket{0,r}$, for different numbers of initial copies $N$ used in the protocol. The success probability of the protocol is shown in the inset. In most regions, more copies lead to more successful performance even though the success probability decreases significantly as $N$ increases. However, for smaller $d/\sigma$ more copies does not help to improve the performance. {\bf Bottom}: The variance in the $\hat{x}$-quadrature for the final state as a function of $d/\sigma$. The protocol is successful when the final state has a smaller variance than the one of the vacuum state (dotted grey line). Optimal squeezing (dash-dotted grey line) is the $\hat{x}$-variance of $\ket{0,r}$. For the simulation we set the squeezing parameter $r=0.7$ and the probability of random displacement to $p=0.5$.}
\label{fig:ManyCopy_SQdist}
\end{figure}

The setting, shown in the bottom panel of Fig.~\ref{fig:OneShot_SQdist}, is similar to the one used for one-shot squeezing distillation. However, we are now allowed to use $N$ independent and identically distributed (iid) copies of the initial state $\rho_{\text{in}}$, given in Eq.~\eqref{eq:OS_SQ_initialstate}. In the first step of the protocol, two copies of the system are considered, and fed to a $50:50$ beam splitter (that is, one with transmissivity $t=1/2$). Indeed, since both states are composed by a mixture of the same squeezed state $\ket{0,r}$ (which has been displaced or not), the beam splitter does not degrade the amount of squeezing contained in the two copies. After the beam splitter, the two systems are classically correlated, and the global state is
\begin{align}
\rho_{12}' =\; &(1-p)^2 \, \ket{0,r}\!\!\bra{0,r}_1 \otimes \ket{0,r}\!\!\bra{0,r}_2 \nonumber \\ 
&+ p (1-p) \, \ket{d',r}\!\!\bra{d',r}_1 \otimes \ket{d',r}\!\!\bra{d',r}_2 \nonumber \\ 
&+ p (1-p) \, \ket{-d',r}\!\!\bra{-d',r}_1 \otimes \ket{d',r}\!\!\bra{d',r}_2 \nonumber \\ 
&+
(1-p)^2 \, \ket{0,r}\!\!\bra{0,r}_1 \otimes \ket{2 d',r}\!\!\bra{2 d',r}_2\,,
\end{align}
where $d' = \frac{d}{\sqrt{2}}$. From the above equation, we can see that by measuring the displacement of the second state, we can infer the one of the first state.

The following step consists of a homodyne measurement performed on the second copy (that is playing the role of the pointer in this setting), followed by post-selection. This measurement is composed by a dichotomic POVM $\left\{ \Pi, \id - \Pi \right\}$ analogous to the one described in the previous section. In particular, the effect is equal to the one in Eq.~\eqref{eq:sq_effect}, with the main difference that the cut-off region is now given by $(-\Delta', \Delta']$, where $\Delta' = \exp(-r)$. If the outcome of the measurement is associated with the effect $\Pi$ the protocol is successful, since with high probability the remaining system is described by $\ket{0,r}$. Otherwise, the protocol fails and the remaining system is discarded. The protocol can be iterated over multiple copies, obtaining a final state that approaches the target state $\ket{0,r}$ at the expenses of a decreasing success probability (exponentially small in the number of iterations performed). An interesting open question left to address is whether the displacement noise can be corrected for all values of $d$.

The numerical simulation is shown in Fig.~\ref{fig:ManyCopy_SQdist}. Again, we plot the fidelity for $\ket{0,r}$ as well as the variance of $\hat{x}$-quadrature of the final state as a function of $d/\sigma$ for different numbers of copies $N$. The fidelity graph also shows the success probability in the inset. We observe that the performance is significantly improved compared to the one-shot case shown in Fig.~\ref{fig:numerics_sq_single}: The fidelity is higher than that of the vacuum state (dotted grey line) in most regions of $d/\sigma$; and the final state can reach fidelity 1 for large $d/\sigma$, which was not possible in the one-shot case. This second point means that the outcome state is exactly $\ket{0,r}$, which is the most resourceful state in the initial mixture, and thus proves that the upper bound in Cor.~\ref{cor:NewNGTheorem} can be saturated so that it is a tight bound. However, the success probability decreases significantly as $N$ and $d/\sigma$ increase. Again, no matter what the values of $N$ and $d/\sigma$ are, the final state can never be more squeezed than the optimal squeezing (dash-dotted grey line in the variance graph) as stated in Cor.~\ref{cor:NewNGTheorem}. In most regions larger $N$ results in a better performance, but for small $d/\sigma$ using more copies improves neither the $\hat{x}$-variance nor the fidelity. This is because when the two squeezed states in the initial mixture have a large overlap, the post-selection procedure does not filter out the unwanted parts as effectively, and more iterations just end up with more unwanted elements in the final state.


\subsection{Entanglement distillation}

\begin{figure}[t!]
\centering
\includegraphics[width=0.7\columnwidth]{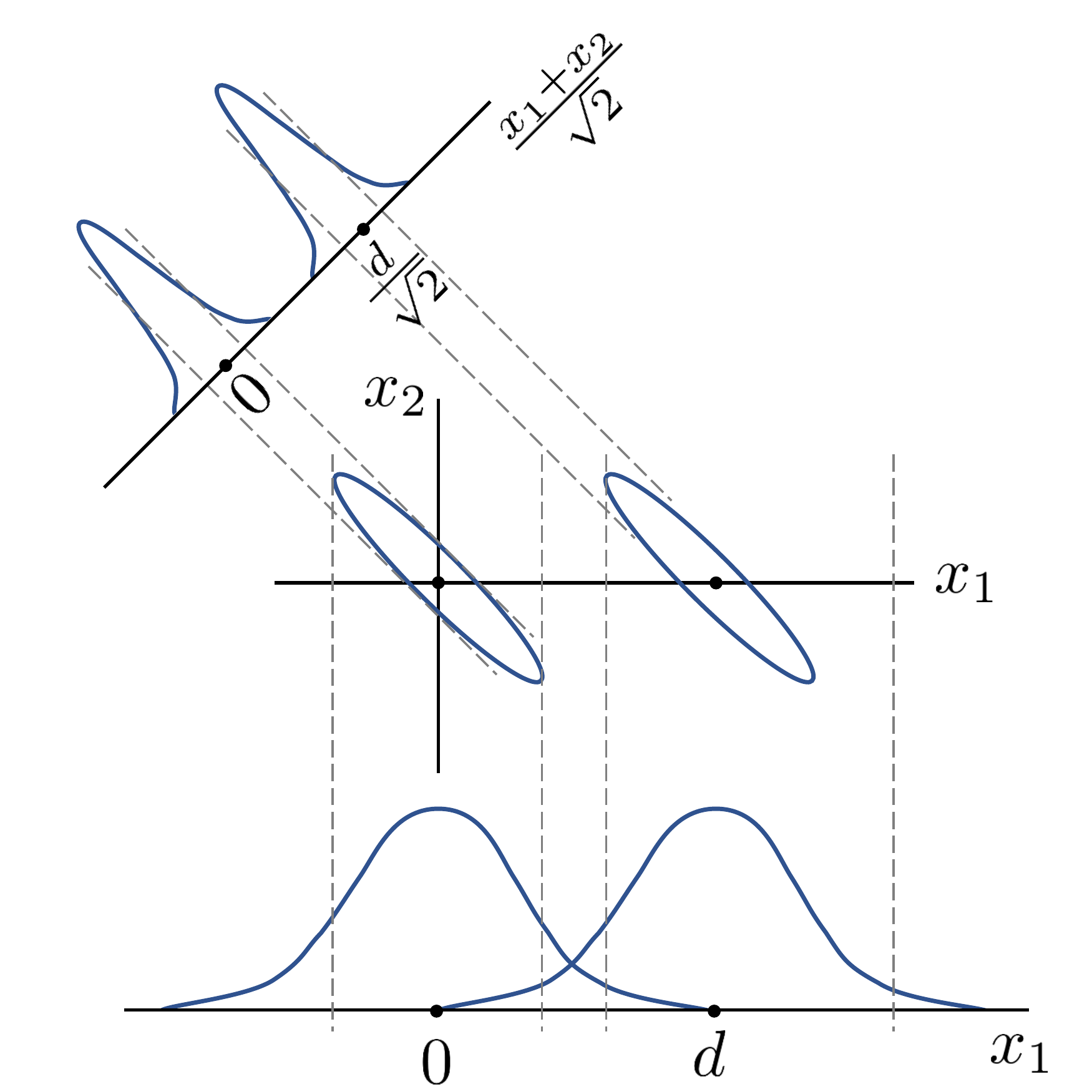}
\caption{The marginal probability distributions obtained from the Wigner function of the state $\rho_{\text{in}}$, given in Eq.~\eqref{eq:MC_ENT_initialstate}. In the $x_1$-$x_2$ diagram, the distribution is represented by two Gaussians squeezed along the $x_1+x_2$ direction, one centred in the origin and the other displaced by $x_1 = d$. If Bob was to measure his local quadrature $x_1$ alone, the two Gaussians would overlap and distinguishing then would be more difficult than measuring the non-local quadrature $x_1+x_2$, where we see that the two distributions have a smaller overlap.}
\label{fig:ent_state}
\end{figure}

We now study entanglement distillation in convex Gaussian resource theories. We show that, when convex mixtures of Gaussian states and conditional Gaussian LOCC operations are allowed, entanglement can indeed be distilled. This is achieved with a protocol similar to the one introduced for the squeezing resource theory and works in the multi-copy probabilistic case.

Consider the following scenario. Alice prepares pure two-mode squeezed states, by applying to the vacuum state the operator $S_2(r)$ introduced in Eq.~\eqref{eq:two_mode_sq}, and send one mode via a noisy channel to Bob. The noisy channel either applies a finite displacement along a given quadrature with probability $p$, or it leaves the mode unchanged. This noise model is a convenient idealisation that produces a non-Gaussian mixture of Gaussian states, and entanglement can be distilled from it. However, the noise model is not very realistic, since for continuous variables we would expect the displacement noise to be distributed according to some continuous distribution, rather than a discrete one. Since the main goal of this section is to provide an example to entanglement distillation using convex Gaussian resources, we do not worry too much about the physicality of the noise model, but it remains an interesting question whether entanglement could be distilled under a more realistic model.

The state shared by Alice and Bob after the noisy channel is
\begin{align}\label{eq:MC_ENT_initialstate}
\rho_{\text{in}} = (1-p) \ket{\tmsv_r}\!\!\bra{\tmsv_r}_{12} + p \, D_1(d) \ket{\tmsv_r}\!\!\bra{\tmsv_r}_{12} D_1(d)^{\dagger}\,,\quad
\end{align}
where the displacement $D_1(d)$ acts over the (local) quadrature $\hat{x}_1$, while the TMSV state $\ket{\tmsv_r} = S_2(r) \ket{0,0}$ is squeezed along the (non-local) quadratures $\hat{x}_+ = (\hat{x}_1+\hat{x}_2)/\sqrt{2}$ and $\hat{p}_- = (\hat{p}_1-\hat{p}_2)/\sqrt{2}$ , with squeezing parameter $r$. A convenient way of visualising the above state is in the $x_1$-$x_2$ space, shown in Fig.~\ref{fig:ent_state}. In particular, it is clear that the displacement in the $\hat{x}_1$-quadrature can be better detected by measuring the non-local quadrature $\hat{x}_+$ rather than $\hat{x}_1$ itself. This observation informs the choice of protocol we introduce, which is analogous to the one used for squeezing.

The protocol, shown in Fig.~\ref{fig:ManyCopy_ENTdist_DIA}, is the following. Alice and Bob use two copies of the state $\rho_{\text{in}}$ at the time, and feed their local modes into a $50:50$ beam splitter. This operation creates classical correlations between the two systems, and produces the following state
\begin{widetext}
\begin{align}
\rho' =\;& (1-p)^2 \ket{\tmsv_r}\!\!\bra{\tmsv_r}_{12} \otimes \ket{\tmsv_r}\!\!\bra{\tmsv_r}_{34} + p (1-p) \, D_1(d') \ket{\tmsv_r}\!\!\bra{\tmsv_r}_{12} D_1(d')^{\dagger} \otimes D_3(d') \ket{\tmsv_r}\!\!\bra{\tmsv_r}_{34} D_3(d')^{\dagger} \nonumber \\ 
&+ p (1-p) \, D_1(d') \ket{\tmsv_r}\!\!\bra{\tmsv_r}_{12} D_1(d')^{\dagger} \otimes D_3(-d') \ket{\tmsv_r}\!\!\bra{\tmsv_r}_{34} D_3(-d')^{\dagger} + p^2 \, D_1(2 d') \ket{\tmsv_r}\!\!\bra{\tmsv_r}_{12} D_1(2 d')^{\dagger} \otimes \ket{\tmsv_r}\!\!\bra{\tmsv_r}_{34}\,,
\end{align}
\end{widetext}
where $d' = \frac{d}{\sqrt{2}}$, and modes $1$ \& $3$ are held by Bob while modes $2$ \& $4$ by Alice. Alice and Bob now need to identify, by measuring the system on the modes $1$ and $2$, whether this system was displaced or not. Since the displacement was performed along the $\hat{x}_1$ quadrature, one possible way of detecting it is for Bob to measure this quadrature locally. However, this might not be the optimal strategy, since the state is squeezed along the quadrature $\hat{x}_+$, and therefore it looks thermal (i.e., it has a broad variance) on $\hat{x_1}$, see Fig.~\ref{fig:ent_state}. Instead, Alice and Bob can locally measure $\hat{x}_2$ and $\hat{x}_1$ respectively, communicate their outcome classically, and process them. In this way, they are able to measure the non-local quadrature $\hat{x}_+$. This measurement can give more information on the displacement of the state, since the system is squeezed along this quadrature. Before introducing the POVM that Alice and Bob perform, it is worth noting that a system displaced by $\delta$ along $\hat{x}_1$ results displaced by $\frac{\delta}{\sqrt{2}}$ along $\hat{x}_+$. If the squeezing is significant, however, measuring the non-local quadrature still provides an advantage to measuring $\hat{x}_1$. Alice and Bob can then measure a dichotomic POVM $\left\{ \Pi, \id - \Pi \right\}$ on the $1$ \& $2$ modes, given by the following effect
\begin{align}\label{eq:nonlocal_effect}
\Pi =  \int \rmd x_1 \rmd x_2 \,
\chi_{(-\Delta',\Delta']}\left(\frac{x_1+x_2}{\sqrt{2}}\right)
\ket{x_1}\!\!\bra{x_1}_1 \otimes \ket{x_2}\!\!\bra{x_2}_2
\end{align}
where the indicator function $\chi_{(-\Delta',\Delta']}\left(\frac{x_1+x_2}{\sqrt{2}}\right)$ is equal to $1$ for $\Delta' <(x_1+x_2)/\sqrt{2} \leq \Delta'$, and $0$ otherwise, and $\Delta'=\exp(-r)$. It can be computed by Alice and Bob after they share their outcomes.

After measuring the above POVM, Alice and Bob can decide whether to continue the protocol or abort it. If the outcome $x_+ \in \left(-\Delta',\Delta'\right]$, they keep the system in the modes $3$ \& $4$, otherwise they discard the system and start again. Even in the case of success, the final state of the system is composed by a mixture, where the state $\ket{\tmsv_r}$ occurs with high probability, while its displaced versions (by $\pm d'$ along $\hat{x}_1$) have a much lower weight. However, Alice and Bob can iterate the protocol described above, by suitably changing the interval of acceptance in Eq.~\eqref{eq:nonlocal_effect} at each round, to increase the weight of the desired state $\ket{\tmsv_r}$. Clearly, since this protocol is non-deterministic, the probability of success quickly decreases with the number of iterations. It should be stressed that Alice and Bob could keep the protocol fully deterministic, by measuring more intervals along $\hat{x}_+$ and displacing the state depending on the outcome. While this is possibly as efficient as the presented protocol, it is less straightforward since at each iteration the number of states in the mixture increases and all of them need to be displaced.

\begin{figure}[!t]
\includegraphics[width=0.9\columnwidth]{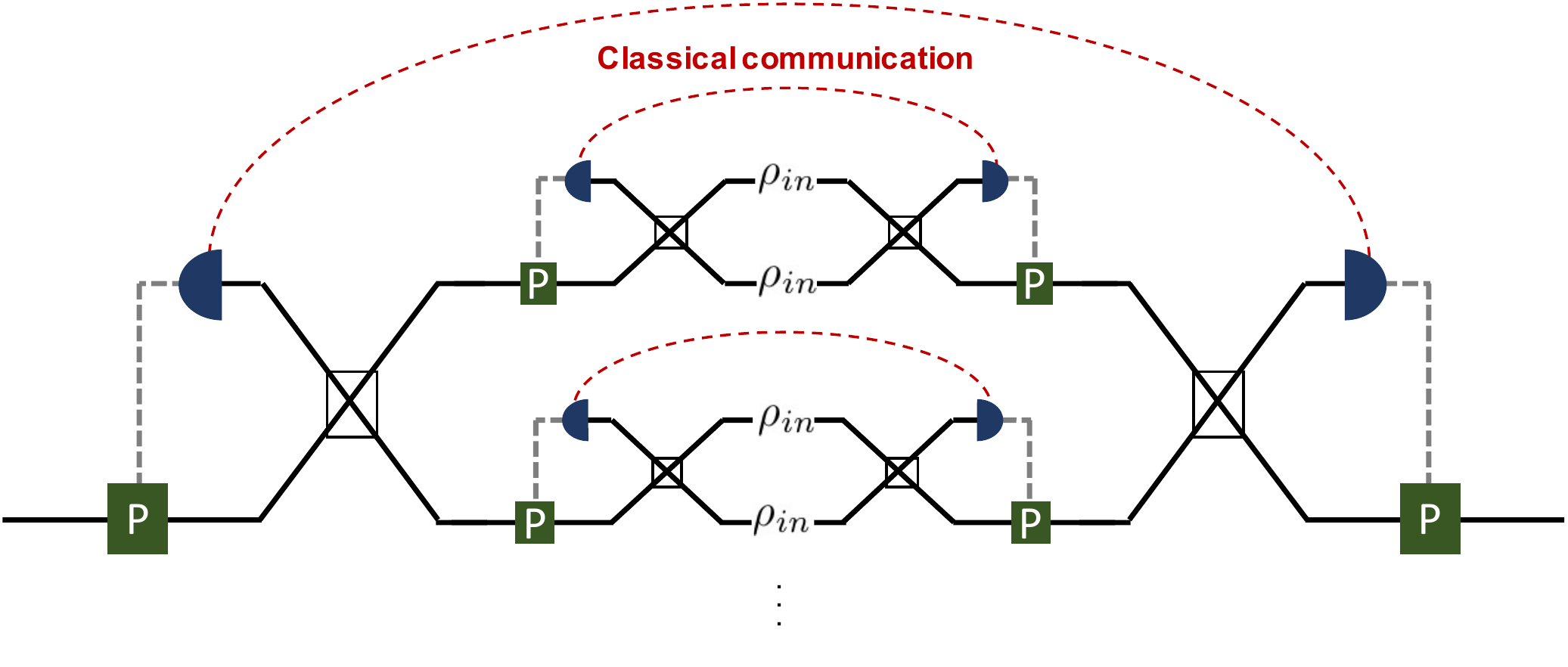}
\caption{Schematic of the probabilistic multi-copy entanglement distillation protocol. Alice and Bob share two copies of the initial state $\rho_{\text{in}}$ in Eq.~\eqref{eq:MC_ENT_initialstate} and feed their local modes into 50:50 beam splitters. Then, via classical communication they perform joint homodyne measurement on the quadrature $\hat{x}_+=(\hat{x}_1+\hat{x}_2)/\sqrt{2}$. If the outcome of the measurement is in the desired range, the protocol is successful, otherwise it is aborted. The protocol can be iterated multiple times with the output state of one iteration being the input state of the next one.}
\label{fig:ManyCopy_ENTdist_DIA}
\end{figure}

\begin{figure}[!t]
\includegraphics[width=\columnwidth]{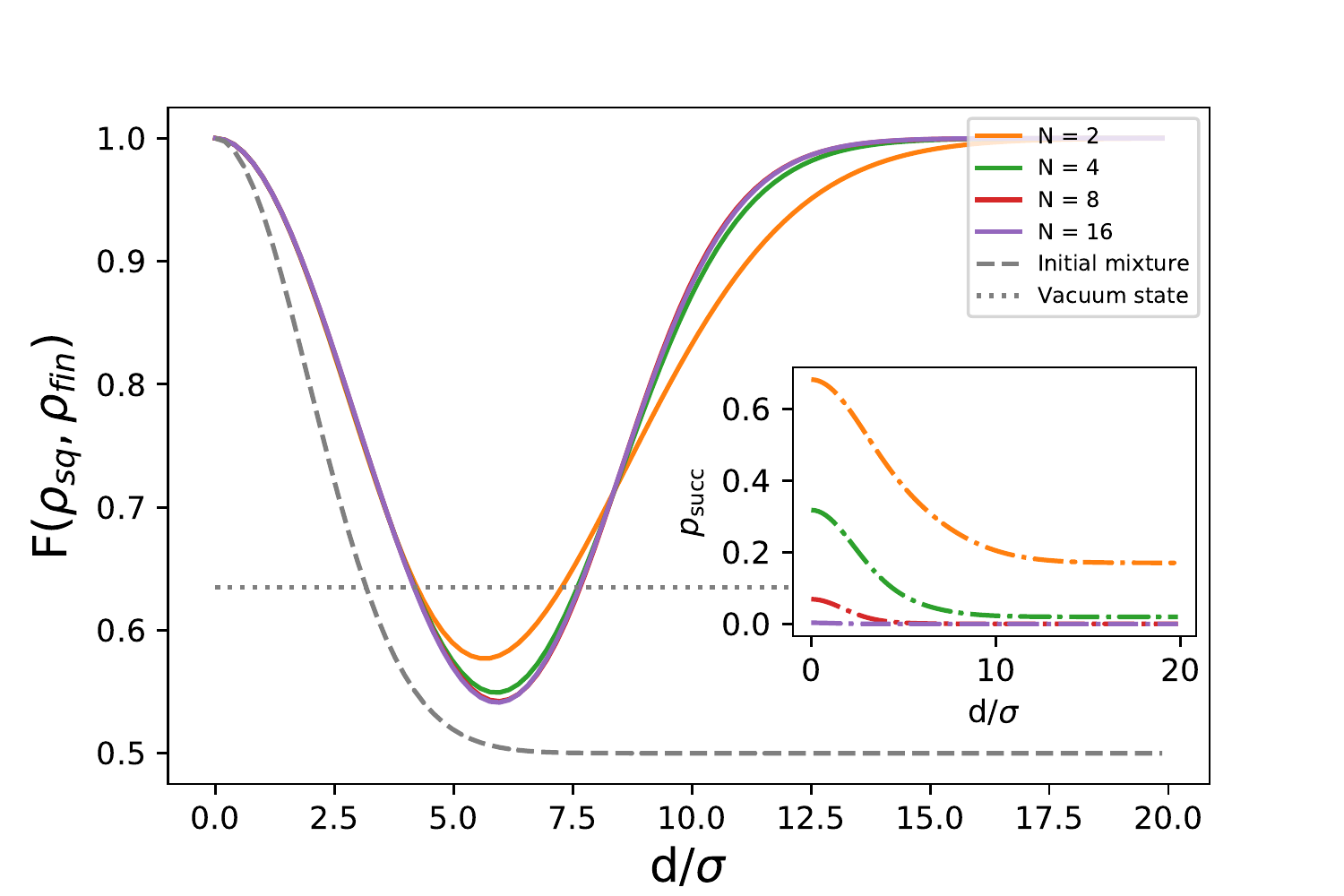}
\caption{Numerical simulation of the multi-copy entanglement distillation protocol for different numbers of copies of the initial state. The graph shows the fidelity between the desired state $\rho_{\text{sq}}=\ket{\psi_r}\!\!\bra{\psi_r}$ and the final state $\rho_{\text{fin}}$ of the protocol as a function of $d/\sigma$. The success probability of the protocol is shown in the inset. The protocol is successful when the final state has a larger fidelity than the initial state as well as the vacuum state (the dotted grey line), which is the closest free state to $\ket{\psi_r}$ \cite{marian2003bures}. The results are similar to the multi-copy squeezing distillation case but generally worse. Using more copies leads to a better performance at the expense of the small success probability for larger $d/\sigma$, but not for small $d/\sigma$. For the simulation, we set the squeezing parameter $r=0.7$, and the probability of random displacement to $p=0.5$.}
\label{fig:ManyCopy_ENTdist}
\end{figure}

The numerical simulation of the protocol, providing the fidelity to $\ket{\psi_r}$ of the output state as well as the success probability of the protocol as a function of $d/\sigma$ for different numbers of copies $N$, is shown in Fig.~\ref{fig:ManyCopy_ENTdist}. The results look similar to the multi-copy squeezing distillation case in Fig.~\ref{fig:ManyCopy_SQdist} as the idea is same, but the performance of the entanglement case is generally worse than the squeezing case. However, still we can achieve fidelity 1 for large $d/\sigma$. As similar to the squeezing case, for small $d/\sigma$ using more copies does not improve the performance since the filtering process via post-selection does not work as effectively for small $d/\sigma$ due to the large overlap between elements in the initial mixture. It is also worth noting that again the upper bound in Cor.~\ref{cor:NewNGTheorem} can be achieved for large $d/\sigma$ (fidelity 1). The protocol described in this section is heavily inspired by the one known as \emph{Gaussification}, introduced in \cite{browne2003driving,eisert2004distillation}, allowing for the distillation of Gaussian entanglement from initial non-Gaussian states. The main difference between the two protocols are due to the fact that we consider a specific class of non-Gaussian states, namely convex mixtures of Gaussian states. This can be leveraged to obtain a more efficient protocol in terms of success probability (indeed, as we stressed before, our protocol can even be made deterministic).


\section{Discussion}\label{sec:discussion}

\paragraph{Summary.} We explored whether classical randomness and conditional operations can help overcoming known limitations to resource distillation in Gaussian resource theories. We introduced the notion of convex Gaussian resource theories, which is motivated by the easy accessibility of convex mixtures of Gaussian resources and their usefulness demonstrated in \cite{takagi2018convex} in the context of distillation of non-Gaussianity. We defined an appropriate resource monotone using a convex roof extension and fully characterised the new limitation on resource distillation in convex Gaussian resource theories. The found limitation does not rule out all resource distillation, but restrict the amount of distillable resources in convex Gaussian resource theories. This seems to indicate that the non-convexity of Gaussian resource theories is not the ultimate obstacle of resource distillation in the Gaussian platform. Nevertheless, we found several distillation protocols for squeezing and entanglement which show that limited resource distillation is possible in convex Gaussian resource theories.

\paragraph{Outlook.} The convex roof extension of the Gaussian resource measure introduced in \cite{mandilara2012quantum} to the convex hull of Gaussian states seems quite general. Historically, the convex roof extension has been used to extend resource measures from pure states to mixed states~\cite{horodecki_quantum_2009}. Even though Gaussian states include mixed states, the relation between Gaussian states and convex mixtures of Gaussian states is not very different from the one between pure states and mixed states. Indeed, any monotone that is well-defined over the set of Gaussian states can be extended to the convex hull of this set by taking the convex roof extension.

One possible application of convex Gaussian resource theories is finding a purification or error correction protocol for non-Gaussian noise. In Sec.~\ref{sec:examples} we demonstrated that convex mixtures of Gaussian states and operations can be employed to correct probabilistic discrete displacement noises. However, as we pointed out before, these noises are not physically realistic for continuous-variable systems, and thus it would be interesting to identify a realistic noise model in which convex Gaussian resources can be useful. An interesting example to look at is non-Gaussian continuous noise. For example, in \cite{franzen2006experimental} the authors studied purification of squeezed states under Gaussian phase fluctuation using Gaussian operations, which fits our setting.

Another open question is whether a similar limitation on resource distillation also holds for states which are closely related to Gaussian states. We mention two specific examples: states with a positive Wigner function and stabilizer states. It is well-known that positive Wigner function states are closely related to Gaussian states. In the pure state case the two sets exactly coincide \cite{hudson1974wigner}. Also, stabilizer states and Clifford operations are commonly considered as the discrete-variable counterpart of Gaussian states and operations \cite{gross2006hudson}. As each set is related to Gaussian states in different ways, this might shine some light on the question addressed in our work: What is the fundamental obstacle of resource distillation in Gaussian formalism and how can we overcome it? One can follow the same procedure with these sets of states, namely defining an appropriate resource theory with new sets of states and operations, and investigate resource distillation in the theory. However, for both theories defining an appropriate resource monotone is the major challenge.


\begin{acknowledgments}
We thank Alessio Serafini, David Jennings, Dan Browne, and Mark Wilde for discussions. H.~H.~J. is supported through a studentship in the Centre for Doctoral Training on Controlled Quantum Dynamics at Imperial College London funded by the EPSRC (EP/L016524/1). C.~S. is funded by EPSRC.
\end{acknowledgments}


\bibliography{apssamp}


\begin{appendix}

\section{Gaussian states and operations}\label{app:Gaussian}

A continuous variable (CV) quantum system has an infinite-dimensional Hilbert space described by canonical operators with continuous eigenspectra obeying the canonical commutation relations \cite{weedbrook2012gaussian}. The most familiar examples are $n$-bosonic mode systems which describe radiation modes of an electromagnetic field. For a given set of $n$ degrees of freedom, the corresponding CV quantum system is associated with a tensor product Hilbert space $\hil^{\otimes n} = \bigotimes_{i=1}^n\hil_i$ and the Hermitian canonical operators $\hat{x}_i$ and $\hat{p}_i$ satisfying the canonical commutation relations
\begin{align}
    [\hat{x}_j,\hat{p}_k] = 2i \, \delta_{j,k} \quad \text{ for } j,k=1, ...\,,
\end{align}
where we set $\hbar=2$ for simplicity. When we define the vector of canonical operators $\hat{\mathbf{x}} = \left(\hat{x}_1,\hat{p}_1, ... ,\hat{x}_n,\hat{p}_n\right)^T$, the canonical commutation relations can be written in a simple form as
\begin{align}
    \left[\hat{x}_j,\hat{x}_k\right] = 2i \, \Omega_{jk}\,,
\end{align}
where $\Omega$ is the \emph{symplectic} matrix
\begin{align}
\Omega = \bigoplus_{i=1}^n \Omega_1
, \quad
\Omega_1=\begin{pmatrix}0&1\\-1&0\end{pmatrix}\,.
\end{align}

The canonical operators have continuous eigenspectra $x\in\mathbb{R}$ and $p\in\mathbb{R}$ and span a real symplectic space called the \emph{phase space}. Any quantum state described by a density operator $\rho$ has another information-complete representation in terms of a quasi-probability distribution in the phase space\,---\,the \emph{phase-space representation}. The most common quasi-probability distribution used in quantum optics is the Wigner function $W(\mathbf{x})$, which is defined as a Fourier transform of the characteristic function of the displacement operator
\begin{align}
    W(\mathbf{x}) = \int_{\mathbb{R}^{2n}} \frac{d^{2n}\mathbf{r}}{(2\pi)^{2n}} \exp\left(-i\mathbf{x}^T\Omega\mathbf{r}\right)\chi(\mathbf{r})
\end{align}
for $\mathbf{x}\in\mathbb{R}^{2n}$, where $\chi(\mathbf{r}) = \tr\left[\rho D(\mathbf{r})\right]$ is the characteristic function and $D(\mathbf{r}) = \exp\left(i\hat{\mathbf{x}}^T\Omega\mathbf{r}\right)$ is the displacement operator.

Gaussian states, $\mathcal{G}$, are those states with a Gaussian quasi-probability distribution. Since Gaussian distributions are fully described by their mean and variance, Gaussian states can be represented solely in terms of their first and second statistical moments of $\hat{\mathbf{x}}$. The first moment vector of a state $\rho$ is defined as
\begin{align}
    \bar{\mathbf{x}} := \expval{\hat{\mathbf{x}}} = \left( \expval{\hat{x}_1}, \expval{\hat{p}_1},...,\expval{\hat{x}_n},\expval{\hat{p}_n}\right)^T
\end{align}
and the second moment matrix (covariance matrix) $V$ is defined as
\begin{align}
    V_{ij} := \frac{1}{2}\expval{\left\{\Delta\hat{x}_i,\Delta\hat{x}_j\right\}} = \frac{1}{2}\expval{\hat{x}_i\hat{x}_j + \hat{x}_i\hat{x}_j} - \expval{\hat{x}_i}\expval{\hat{x}_j}\,.
\end{align}
The Heisenberg uncertainty principle impose the following constraints on the covariance matrix
\begin{align}\label{eq:unc_rel}
V + \frac{i}{2} \Omega \geq 0\,.
\end{align}
There is a one-to-one mapping between the set of Gaussian states and the set of pairs of first moment vectors and covariance matrices $\left[\bar{\mathbf{x}},V\right]$, and therefore we can denote any Gaussian state as $\rho^G[\bar{\mathbf{x}},V]$. When the first moment is not important to the context, we use the notation $\rho^G[V]$. This compact description of Gaussian states is the reason why Gaussian states are mathematically easier to work with than general infinite-dimensional CV states. The Wigner function of a Gaussian state $\rho^G[\bar{\mathbf{x}},V]$ can be written as
\begin{align}
    W(\mathbf{x}) = \frac{\exp\left[-\left(\mathbf{x}-\bar{\mathbf{x}}\right)^TV^{-1}\left(\mathbf{x}-\bar{\mathbf{x}}\right)\right]}{\left(\pi/2\right)^N\sqrt{\det V}}\,.
\end{align}
Gaussian operations, $\mathbb{G}$, are quantum operations which preserve Gaussianity of Gaussian states. They correspond to operations realised by Hamiltonians which are quadratic in the canonical operators \cite{serafini2017quantum}.

Gaussian states are easy to produce and control in the laboratory. Familiar examples of Gaussian states are the vacuum state, coherent states, thermal states, and squeezed states. Since in the main text we largely make use of squeezed states, we briefly recall their definitions and properties in this appendix. Squeezed states saturate the uncertainty relations in Eq.~\eqref{eq:unc_rel} but, unlike coherent states, one can find a pair of quadratures for which one of the variances is minimised while the other one is maximised. This feature makes squeezed states resourceful in metrology, and more in general in those tasks requiring high sensitivity in a given quadrature.

For a single mode, the \emph{squeezed coherent state} is defined as 
\begin{align}\label{eq:SqCS}
\ket{\alpha,r} = D(\alpha) S(r) \ket{0}\,,
\end{align}
where $\ket{0}$ is the vacuum state, $D(\alpha)$ is the displacement operator, and
\begin{align}
    S(r) = \exp\left\{ \frac{r}{2} \left(\hat{a}^2-\left(\hat{a}^{\dagger}\right)^2\right) \right\}
\end{align}
is the squeezing operator, with $\hat{a} = \frac{1}{2}\left(\hat{x}+i\hat{p}\right)$ the bosonic field operator. The covariance matrix of this state is
\begin{align}
V_{\smsv}(r)
=
\begin{pmatrix}
e^{-2r} & 0 \\
0 & e^{2r}
\end{pmatrix}\,,
\end{align}
where it is easy to see that the $\hat{x}$ quadrature has been squeezed at the expense of the $\hat{p}$ quadrature. Another relevant squeezed state is the one for two modes, since this state contains some entanglement, and is the Gaussian equivalent of a Bell state for an amount of squeezing tending to infinity. The two-mode squeezed vacuum (TMSV) state is obtained from the vacuum $\ket{0}$ by applying the operator
\begin{align}\label{eq:two_mode_sq}
S_2(r)
=
\exp\left\{-\frac{r}{2} \left(\hat{a}_1\hat{a}_2 - \hat{a}_1^{\dagger}\hat{a}_2^{\dagger}\right) \right\}\,.
\end{align}
This state is squeezed in the $\hat{x}_1+\hat{x}_2$ and $\hat{p}_1-\hat{p}_2$ quadratures at the expenses of the large variances in $\hat{x}_1-\hat{x}_2$ and $\hat{p}_1+\hat{p}_2$.


\section{Monotones for squeezing}\label{app:var_is_monotone}

In the main text we use, together with the fidelity to the target state, the variance with respect to a specific quadrature as a measure for squeezing. In this appendix, we aim to justify its use over convex mixtures of Gaussian states. Let us first introduce the measure
\begin{align}\label{eq:variance_measure}
M_{\text{var}}(\rho)
=
\min_{\hat{Q}}
\text{Var} \left[ \hat{Q} \right]_{\rho}\,,
\end{align}
where the minimization is performed over all quadrature operators $\hat{Q}$, and the variance is defined as
\begin{align}
\text{Var} \left[ \hat{Q} \right]_{\rho}
=
\langle \hat{Q}^2 \rangle_{\rho} - \langle \hat{Q} \rangle_{\rho}^2\,,
\end{align}
where $\langle \hat{A} \rangle_{\rho} = \Tr \left[ \hat{A} \, \rho \right]$. It is worth noting that the above measure coincide with the minimum eigenvalue of the covariance matrix of the state $\rho$, a squeezing measure for pure Gaussian states~\cite{kraus2003entanglement}. The above measure is monotone under the allowed operations of squeezing theory $\mathbb{O}_{\text{sq}}$, which are composed by the following fundamental operations:
\begin{enumerate}[(i)]
\item Appending modes in the vacuum state
\item Performing passive operations (beam splitter, phase shift, and displacements)
\item Trace out a subset of the system's modes.
\end{enumerate}
When the state space is extended to the set of all Gaussian states, the measure in Eq.~\eqref{eq:variance_measure} can be modified as
\begin{align}\label{eq:variance_measure_mixed}
\overline{M}_{\text{var}}(\rho)
=
\min \left\{ 1, M_{\text{var}}(\rho) \right\}\,.
\end{align}
Notice that we have to modify the measure since $M_{\text{var}}$ is not monotonic over states whose covariance matrix has all eigenvalues higher than $1$. This happens for instance in the case of thermal states). Indeed, one could always decrease the measure by replacing these states with the vacuum. To avoid this problem, the minimization in Eq.~\eqref{eq:variance_measure_mixed} has been added. Finally, it is worth noting that in the main text we do not explicitly consider $\overline{M}_{\text{var}}$, but rather we focus on the relevant quadrature for the problem.

From the examples in Sec.~\ref{sec:examples} of the main text, it is clear that $\overline{M}_{\text{var}}$ is not a monotone under the allowed operations of Def.~\ref{def:all_op_cgt}. However, we would like to argue that this measure is still a meaningful one over the set of convex mixtures of Gaussian states. To do so we take a different notion of allowed operations, that are still relevant for the state space we are considering, but do not allow for conditioned operations.\footnote{These are the ones allowing for a reduction of the variance in our examples in Sec.~\ref{sec:examples}.} The operations we consider here are mixtures of allowed operations for squeezing theory. The most general map in this set can be written as
\begin{align}\label{eq:all_op_sq_convex}
\Gamma(\rho)
=
\int \rmd \lambda \, p(\lambda) \, \Phi_{\lambda}(\rho)\,,
\end{align}
where $p(\lambda)$ is a probability distribution, and $\Phi_{\lambda} \in \mathbb{O}_{\text{sq}}$.

Let us first show that the measure $\overline{M}_{\text{var}}$ is monotone under $\mathbb{O}_{\text{sq}}$ alone, when the state space is given by mixtures of Gaussian states. If the mixture of Gaussian states, let us call it $\rho_{\text{mix}}$, is such that $M_{\text{var}}(\rho_{\text{mix}}) < 1$, then it is easy to see that neither {\it (i)} appending ancillary systems in the vacuum state, nor {\it (ii)} performing passive operations modify this value. Indeed, passive operations are represented by orthogonal matrices acting over the covariance matrix by congruence (irrespectively of whether the state is Gaussian or not), thus preserving the eigenvalues of the matrix. Furthermore, {\it (iii)} partial tracing the system corresponds to selecting a principal sub-matrix of the covariance matrix and discarding the rest. It is known~\cite{horn_matrix_1985} that the smallest eigenvalue of the sub-matrix cannot be lower than the original matrix. Thus, when $M_{\text{var}}(\rho_{\text{mix}}) < 1$, we have that
\begin{align}\label{eq:monotonic_var}
M_{\text{var}}\left(\Phi(\rho_{\text{mix}}) \right)
\geq
M_{\text{var}}\left(\rho_{\text{mix}} \right)\,,
\end{align}
where $\Phi \in \mathbb{O}_{\text{sq}}$. When $M_{\text{var}}(\rho_{\text{mix}}) \geq 1$, instead, it is easy to see that only operation {\it (i)} can decrease the variance, since the vacuum's covariance matrix is $\id$. However, the variance cannot be reduced further than the unit, so that $\bar{M}_{\text{var}}$ is monotonic under the allowed operations $\mathbb{O}_{\text{sq}}$ over mixtures of Gaussian states.

Monotonicity of this measure under maps of the form given in Eq.~\eqref{eq:all_op_sq_convex} follows from the concavity of the variance. When the output state of the channel has a minimum variance that is smaller than unit, we have that
\begin{align}
\bar{M}_{\text{var}} \left( \Gamma ( \rho_{\text{mix}} ) \right)
&=
\text{Var} \left[ \bar{Q} \right]_{\Gamma ( \rho_{\text{mix}} )} \\
&\geq
\int \rmd \lambda \, p(\lambda) \, \text{Var} \left[ \bar{Q} \right]_{\Phi_{\lambda}( \rho_{\text{mix}} )} \\
&\geq
\int \rmd \lambda \, p(\lambda) \, \bar{M}_{\text{var}} \left( \Phi_{\lambda}( \rho_{\text{mix}} ) \right) \\
&\geq
\int \rmd \lambda \, p(\lambda) \, \bar{M}_{\text{var}} \left( \rho_{\text{mix}} \right)
=
\bar{M}_{\text{var}} \left( \rho_{\text{mix}} \right)\,,
\end{align}
where $\bar{Q}$ is the quadrature operator minimising the variance of $\Gamma ( \rho_{\text{mix}} )$, and the first inequality follows from the concavity of the variance and by the definition of $\Gamma$. The second inequality follows from the fact that $\bar{M}_{\text{var}}$ involves a minimisation over all quadrature operators (plus a cut-off for values higher than the unit), and the last inequality follows from Eq.~\eqref{eq:monotonic_var}. When the minimum variance of the output state is already higher than unit, the monotonicity relation is trivially true. Thus, variance is a good measure of squeezing over convex mixtures of Gaussian states, although it is not monotonic when the allowed operations include conditional maps.


\section{Proof of Lem.~\ref{lem:PropertiesTkappa}}\label{app:prop_tkappa_proof}

In this appendix, we provide a proof for the properties of the resource measure $\tkappa$ as stated in Lem.~\ref{lem:PropertiesTkappa} in the main text and recalled here for convenience.

\proptk*

\begin{proof}
{\it(i)} Consider a Gaussian state $\rho^G[V]$ with covariance matrix $V$. Any convex decomposition of this state in terms of Gaussian states $\left\{ \rho^G(\lambda) \right\}_{\lambda}$ has the form
\begin{align}
\rho^G[V] = \int \rmd \lambda \, p_G(\lambda) \, \rho^G(\lambda)\,.
\end{align}
where the probability distribution $p_G$ is Gaussian (if the state is mixed) or a Dirac delta (if the state is pure). Furthermore, each state $\rho^G(\lambda)$ has a covariance matrix $W(\lambda)$ such that $W(\lambda)\leq V$ \cite[Lem.~3]{wolf2004gaussian}. This implies that $\kappa(\rho^G[V])\leq\kappa(\rho^G(\lambda))$ for all $\rho^G(\lambda)$, for any Gaussian decomposition. Then, given the optimal decomposition $\left\{ \bar{p}(\lambda),\bar{\rho}^G(\lambda) \right\}_{\lambda}$ for the measure $\tkappa$, we have that
\begin{align}
\tkappa\left(\rho^G\right)
&=
\int \rmd\lambda\,
\bar{p}(\lambda) \,
\kappa\left(\bar{\rho}^G(\lambda)\right) \nonumber \\
&\geq
\int \rmd\lambda\,
\bar{p}(\lambda) \,
\kappa\left(\rho^G[V]\right)
=
\kappa\left(\rho^G[V]\right)\,.
\end{align}
However, equality in the above equation can be achieved by $\rho^G$ itself, so $\tkappa(\rho^G[V]) = \kappa(\rho^G[V])$.

{\it(ii)} Consider the state $\tau = \int \rmd\lambda\, p(\lambda)\rho(\lambda)$. If we denote $\{q(\lambda,\mu),\rho^G(\lambda,\mu)\}_{\mu}$ as the optimal decomposition for each state $\rho(\lambda)$ such that
\begin{align}
    \tkappa(\rho(\lambda)) = \int \rmd\mu\, q(\lambda,\mu) \kappa\left(\rho^G(\lambda,\mu)\right)\,,
\end{align}
the set $\{p(\lambda)q(\lambda,\mu),\rho^G(\lambda,\mu)\}_{\lambda,\mu}$ is one possible decomposition of $\tau$. It then follows that
\begin{align}
    \tkappa(\tau) &\leq \int \rmd\lambda\,\rmd\mu\,p(\lambda)q(\lambda,\mu) \kappa\left(\rho^G(\lambda,\mu)\right) \\
    &= \int \rmd\lambda\, p(\lambda) \tkappa\left(\rho(\lambda)\right)\,.
\end{align}

{\it(iii)} We showed in (i) that for a Gaussian state $\rho^G$ it holds that $\tkappa(\rho^G)=\kappa(\rho^G)$. Using the fact that tensor products of two Gaussian states are Gaussian and the property of $\kappa$ introduced in Sec.~\ref{subsec:GaussianQRT} we have
\begin{align}
\tkappa\left(\rho^G\otimes\sigma^G\right) = \kappa\left(\rho^G\otimes\sigma^G\right) &= \max\Big\{\kappa\left(\rho^G\right),\kappa\left(\sigma^G\right)\Big\}\\&= \max\Big\{\tkappa\left(\rho^G\right),\tkappa\left(\sigma^G\right)\Big\}\,.
\end{align}

{\it (iv)} Monotonicity with respect to appending an ancillary system in a free state follows straightforwardly from property (iii) and from the definition of the measure $\kappa$. For the monotonicity under mixtures of conditional Gaussian operations, Eq.~\eqref{eq:generalForm_convG}, let us consider $\{p(\lambda),\rho_{AB}^G(\lambda)\}_{\lambda}$ to be the optimal decomposition of $\rho_{AB}$ such that $\tkappa(\rho_{AB}) = \int \rmd\lambda \, p(\lambda) \, \kappa(\rho_{AB}^G(\lambda))$. For a general operation $\Gamma(\cdot) = \int \rmd q \, \Phi_A^q \otimes M_B^q \left( \cdot \right)$, we denote the states $\rho_{AB}^G(\lambda)$ after partial selective measurements as 
\begin{align}
\rho^G_A(\lambda|q) = \frac{M_B^q \left( \rho^G_{AB}(\lambda) \right)}{h(q|\lambda)}
\end{align}
where $h(q|\lambda)=\tr\left[\rho^G_{AB}(\lambda)\left(\id_A\otimes\ket{q}\bra{q}\right)\right]$ is the probability of outcome $q$. Then,
\begin{align}
\tkappa\left(\Gamma(\rho_{AB})\right)
&\leq
\int \rmd\lambda \, \rmd q \,
p(\lambda) \, \tkappa
\Big(
\Phi^q_A\otimes M^q_B \left(\rho_{AB}^G(\lambda)\right)
\Big) \\
&=
\int \rmd\lambda \, \rmd q\,
p(\lambda) h(q|\lambda) \,
\tkappa\Big( \Phi^q_A \left( \rho^G_A(\lambda|q)\right)\Big) \\
&= \int \rmd\lambda \rmd q\, p(\lambda) h(q|\lambda) \, \kappa\Big( \Phi^q_A \left( \rho^G_A(\lambda|q)\right)\Big)\,,
\end{align}
where we used properties (ii) and (i) in the first line and the last line, respectively. Then, using the monotonicity of $\kappa$ under $\Phi^q\in\mathbb{O}_G$ and partial selective homodyne measurements (discussed in App.~\ref{app:Part_sel_HM}), we obtain
\begin{align}
\tkappa\left(\Gamma(\rho_{AB})\right) \leq& \int \rmd\lambda \rmd q\, p(\lambda) h(q|\lambda) \, \kappa\left( \rho_{AB}^G(\lambda) \right) \\
=& \int \rmd\lambda\, p(\lambda) \kappa\left( \rho_{AB}^G(\lambda) \right) = \tkappa(\rho_{AB})\,.
\end{align}

{\it (v)} Let us start from the only if direction. Consider a state $\rho\in\convG$ with optimal decomposition $\{p(\lambda),\rho^G(\lambda)\}$ such that
\begin{align}
\tkappa(\rho)
=
\int \rmd \lambda \, p(\lambda) \, \kappa(\rho^G(\lambda))
= 1\,.
\end{align}
By definition $\kappa(\rho) \geq 1$ for all states $\rho \in \mathcal{G}$, and since $\tkappa(\rho)=1$ we have that $\kappa(\rho^G(\lambda))=1$ for all $\lambda$. Due to the faithfulness of $\kappa$, it follows that $\rho^G(\lambda) \in \mathcal{F}_G$ for all $\lambda$. Then, $\rho$ is a convex combination of free states, which implies that
$\rho\in\mathcal{F}_{G^c}$, since $\mathcal{F}$ is convex.

For the if direction, let us assume that $\rho\in\mathcal{F}_{G^c}$. Then, by definition of the set of free states, a decomposition of $\rho$ in terms of free Gaussian states is possible
\begin{align}
\rho = \int \rmd \lambda \, p(\lambda) \, \rho_F^G(\lambda)\,,
\end{align}
where $\rho_F^G(\lambda) \in \mathcal{F}_{G}$ for all $\lambda$. Due to the faithfulness of $\kappa$, it follows that $\tkappa(\rho) \leq 1$, and since the value of $\kappa$ is lower bounded by $1$, we have that $\tkappa(\rho) = 1$.
\end{proof}


\section{Partial selective homodyne measurement}\label{app:Part_sel_HM}

In this appendix, we prove that the partial selective homodyne measurement as introduced within Def.~\ref{def:all_op_cgt} is an allowed operation for the Gaussian resource theories of squeezing and entanglement. The operation, defined as
\begin{align*}
M_B^q \ : \ &\hil_A \otimes \hil_B \rightarrow \hil_A \\
&\rho_{AB} \xmapsto{\phantom{hil_A \otimes \hil_B}} \Tr \left[\id_A \otimes \ket{q}\bra{q}_B \, \rho_{AB}\right]
\end{align*}
consists in the application of the partial projector on the eigenvector $\ket{q}$ of a given quadrature operator $\hat{Q}$, followed by tracing out the measured modes. Notice that the subsystems $A$ and $B$ can both be composed by multiple modes. Even though the quadrature eigenvector $\ket{q}$ is an infinitely squeezed state, one can show that this does not increase the amount of squeezing (or entanglement) in the state. Let us first recall how the partial selective homodyne measurement changes the covariance matrix of a state.

\begin{lemma}\label{lem:state_after_partial_HM}
\cite{giedke2002characterization,eisert2002distilling} Let $V$ be the covariance matrix of a state describing a $(n+m)$-mode system, which can be written as 
\begin{align}
    V = \begin{pmatrix}
    A&B\\B^T&C
    \end{pmatrix}\,,
\end{align}
where $A, B$, and $C$ are real matrices with size $2n\times2n$, $2n\times2m$, and $2m\times2m$ respectively. The covariance matrix $V'$ of the state after a selective homodyne measurement on the last $m$ modes is then given by
\begin{align}\label{eq:Mq_cov_mat}
V' = A - B\left(PCP\right)^{MP}B^T\,,
\end{align}
where $P=\id_m\oplus\mathbf{0}_m$ is a projection operator, and $MP$ denotes the Moore Penrose pseudo-inverse. 
\end{lemma}

Using this lemma, we can show that partial selective homodyne measurement is an allowed operation for both the Gaussian resource theory of squeezing and entanglement.

\begin{proposition}\label{lem:kappa_partial_measurement}
Partial selective homodyne measurements on Gaussian states are an allowed operation for the Gaussian theory of squeezing. That is, we have for all $\rho_{AB} \in \mathcal{F}^{\mathrm{sq}}_G$ that $M_B^q(\rho_{AB}) \in \mathcal{F}^{\mathrm{sq}}_G$.
\end{proposition}

\begin{proof}
Instead of considering the set of free states $\mathcal{F}^{\mathrm{sq}}_G$, we focus on the set of free covariance matrices, characterised as $\mathbf{F}^{sq}_G=\{V \, | \, V\geq\id\}$ for the theory of squeezing. Then, for any $(n+m)$-mode free covariance matrix $V\in\mathbf{F}^{sq}_G$, we can write 
\begin{align}\label{eq:submatrices_of_V}
    V=\begin{pmatrix}
    A&B\\B^T&C
    \end{pmatrix}
    \geq
    \begin{pmatrix}
    \id_{2n}&\mathbf{0}\\
    \mathbf{0}&\id_{2m}
    \end{pmatrix}
    \geq
    \begin{pmatrix}
    \id_{2n}&\mathbf{0}\\
    \mathbf{0}&\mathbf{0}_{2m}
    \end{pmatrix}\,.
\end{align}
The last inequality imply that the following matrix
\begin{align}
W =
\begin{pmatrix}
A- \id_{2n} & B \\
B^T & C
\end{pmatrix}
\end{align}
is positive semi-definite. We can now construct a different matrix by applying the operator $\id_{2n} \oplus P$ to $W$ by congruence, where $P=\id_m\oplus\mathbf{0}_m$.
\begin{align}
\tilde{V} =
(\id_{2n}\oplus P)W(\id_{2n}\oplus P) =
\begin{pmatrix}
A-\id_{2n} & BP \\
PB^T & PCP
\end{pmatrix}\,.
\end{align}
Clearly, this matrix is positive semi-definite as well. Then, its Schur complement must be positive semi-definite,
\begin{align}
\left(A-\id_{2n}\right) - BP(PCP)^{MP}PB^T \geq 0\,.
\end{align}
Furthermore, one can easily show that $P(PCP)^{MP}P$ satisfies all four conditions to be the Moore Penrose pseudo-inverse of $PCP$ \cite{watrous2018theory}. Due to the uniqueness of the pseudo-inverse, this implies that $P(PCP)^{MP}P=(PCP)^{MP}$. Then, we conclude that
\begin{align}
A - B(PCP)^{MP}B^T \geq \id_{2n}\,.
\end{align}
In Lem.~\ref{lem:state_after_partial_HM}, we showed that this is in fact the resulting covariance matrix after a partial selective homodyne measurement is performed on a Gaussian state with covariance matrix $V$. This closes the proof.
\end{proof}

The case of entanglement is easier to deal with since the partial selective homodyne measurement needs to be a local operation. Any separable Gaussian state has a covariance matrix $V_{AB} = V_A \oplus V_B$, where $V_A$ ($V_B$) is the covariance matrix of the state owned by Alice (Bob). Then, if we act with $M^q$ on part of the subsystem of Bob, we obtain a final covariance matrix $V'_{AB} = V_A \oplus V'_B$, where $V'_B$ has the form given in Eq.~\eqref{eq:Mq_cov_mat}, and similarly for Alice. Clearly, the output covariance matrix is that of a separable state.


\end{appendix}

\end{document}